\documentclass[letterpaper,11pt]{article}
\usepackage{chicagor}
\usepackage{verbatim}
\usepackage{xspace}
\usepackage{amsmath}
\usepackage{amsthm}
\usepackage{amssymb}
\usepackage{graphicx}
\usepackage{epstopdf}
\usepackage{url}

\title{The Truth Behind the Myth of the Folk Theorem}
\author{
Joseph Y. Halpern \qquad \qquad Rafael Pass \qquad \qquad  Lior Seeman\\ 
Computer Science Dept.\\
Cornell University\\
Ithaca, NY\\
E-mail: halpern$|$rafael$|$lseeman@cs.cornell.edu
}

\newtheorem{theorem}{Theorem}[section]
\newtheorem{cor}[theorem]{Corollary}
\newtheorem{lemma}[theorem]{Lemma}

\newtheorem{definition}[theorem]{Definition}

\newtheorem*{theorem*}{Theorem}
\newtheorem*{corollary*}{Corollary}
\newtheorem*{conjecture*}{Conjecture}
\newtheorem*{lemma*}{Lemma}
\newtheorem*{thm*}{Theorem}
\newtheorem*{prop*}{Proposition}
\newtheorem*{obs*}{Observation}
\newtheorem*{rem*}{Remark}
\newtheorem*{definition*}{Definition}

\newcommand{\eps}{{\epsilon}}
\newcommand{\Gen}{{\text{Gen}}}
\newcommand{\Enc}{{\text{Enc}}}
\newcommand{\Dec}{{\text{Dec}}}
\newcommand{\PPT}{{\text{PPT}}}

\begin{document}
\date{}
\maketitle

\begin{abstract}
We study the problem of computing an $\epsilon$-Nash equilibrium in
repeated games. 
Earlier work by Borgs et al.~\citeyear{borgs2010myth} suggests that this
problem is intractable.  We show that if we make a slight change to
their model---modeling the players as polynomial-time Turing machines
that maintain state ---and make some standard cryptographic hardness assumptions
(the existence of public-key encryption), the problem can actually be
solved in polynomial time. 
Our algorithm works not only for games with a finite number of players, but
also for constant-degree graphical games. 

As Nash equilibrium is a weak solution
concept for extensive form games, we additionally define and study an appropriate notion of a
subgame-perfect equilibrium for computationally bounded players, and
show how to efficiently find such an equilibrium in repeated games
(again, making standard cryptographic hardness assumptions). 

\end{abstract}

\section{Introduction}
The complexity of finding a Nash equilibrium (NE) is a fundamental
question at the interface of game theory and computer science. A 
celebrated sequence of results showed that the complexity of finding a NE in a
normal-form game is PPAD-complete
\cite{chen2006settling,daskalakis2006complexity}, even for 2-player
games.  
Less restrictive concepts, such as \mbox{$\epsilon$-NE} for an
inverse-polynomial $\epsilon$, are just as hard
\cite{chen2006computing}.  This suggests that these problems are
computationally intractable.

There was some hope that the situation would be better in
infinitely-repeated games.  
The \emph{Folk Theorem} (see \cite{OR94} for a review) informally states that 
in an infinitely-repeated game $G$, 
for any payoff profile that is \emph{individually rational}, in that all 
players get more than\footnote{For our results, since we consider
$\epsilon$-NE, we can replace ``more than'' by ``at least''.} their minimax
payoff  (the highest payoff 
that a player can guarantee himself, no matter what the other players do)
and is the outcome of some correlated strategy in $G$,
there is a Nash equilibrium of $G$ with this payoff profile.
With such a large set of equilibria, the hope was that finding one would
be less difficult.
Indeed, Littman and Stone  \citeyear{littman2005polynomial} showed that
these ideas can be used to design an algorithm for finding a NE in a
two-player repeated game.  

Borgs et al. \citeyear{borgs2010myth} (BC+ from now on) proved some
results suggesting that, 
for more than two players,
even in infinitely-repeated games it would be
difficult to find a NE.  Specifically, they showed that, under certain
assumptions, the problem of finding a NE (or even an $\epsilon$-NE for an
inverse-polynomial $\epsilon$) in an infinitely repeated game with three
or more players where there is a discount factor bounded away
from $1$ by an inverse polynomial is also PPAD-hard. 
They prove this by showing that, given
an arbitrary normal-form game $G$ with $c \ge 2$ players, there is a
game $G'$ with $c+1$ players 
such that finding an $\epsilon/8c$-NE for the repeated game based on
$G'$ is equivalent to finding an $\epsilon$-NE for $G$.

While their proof is indeed correct, in this paper, we challenge their
conclusion.  
If we take seriously the importance of being able to find an
$\epsilon$-NE efficiently, it is partly because we have computationally bounded
players in mind.  But then it seems reasonable to see what happens if we
assume that the players in the game are themselves
computationally bounded. 
Like BC+, we assume that players are
resource bounded.%
\footnote{Although BC+ do not discuss modeling players in this way, the
problem they show is NP-Hard is to find a polynomial-time TM 
profile that implements an equilibrium. There is an obvious exponential-time TM profile that
implements an equilibrium: each TM in the profile just computes the
single-shot NE and plays its part repeatedly.}
Formally, we view players as probabilistic%
\footnote{BC+ describe their TMs as deterministic, but allow them to
output a mixed 
strategy. As they point out, there is no difference between this
formulation and a probabilistic TM that outputs a specific action;
their results hold for such probabilistic TMs as well.} 
polynomial-time 
Turing machines (PPT TMs).   
We differ from BC+ in two key respects.
First, since we restrict to (probabilistic) polynomial-time players, we restrict the deviations that can be made in equilibrium to those that can be computed
by such players; BC+ allow arbitrary deviations.
Second, BC+ implicitly assume that players have no 
memory: they cannot remember computation from earlier rounds.  By way of contrast, we allow players 
to have a bounded (polynomial) amount of memory.
This allows players to remember the results of a few coin tosses from
earlier rounds, and means that we can use some cryptography (making some
standard cryptographic assumptions) to try to coordinate the players.  
We stress that this coordination happens in the process of the game
play, not through communication.  
That is, there are no side channels; the only form of ``communication''
is by making moves in the game.
We call such TMs \emph{stateful}, and the BC+ TMs \emph{stateless}.
We note, that without the restriction on deviations, there is no real difference between
stateful TMs and stateless TMs in our setting (since a player with unbounded computational power can
recreate the necessary state).
With these assumptions (and the remaining assumptions of the BC+ model), we show
that in fact an $\epsilon$-NE in an infinitely-repeated game
can be found in polynomial time.

Our equilibrium strategy uses threats
and punishment much in the same way that they are used in the Folk
Theorem. However, since the players are computationally bounded we can
use cryptography (we assume the existence of a secure public key
encryption scheme) to secretly correlate the punishing players. This
allows us to overcome the difficulties raised by BC+. 
Roughly speaking, the $\epsilon$-NE can be described as proceeding in
three stages.  
In the first stage, the players play a sequence of predefined actions
repeatedly. If some player 
deviates from the sequence, the second stage begins, in which the other
players use their actions to secretly exchange a random seed, through the use of public-key encryption.
In the third stage, the
players use a correlated minimax strategy to punish 
the
deviator forever.
To achieve this correlation, the players use the
secret random seed as the seed of a pseudorandom function, and use the
outputs of the pseudorandom function as the source of randomness for the
correlated strategy. 
Since the existence of public-key encryption implies the existence of pseudorandom
functions, the only cryptographic assumption needed is the existence
of public-key encryptions---one of the most basic cryptographic
hardness assumptions.

In the second part of the paper we show how to extend this result to a more refined solution concept. While NE has some attractive features, it allows some unreasonable
solutions.  In particular, the
equilibrium might be obtained by what are arguably empty threats.  This
actually happens in our proposed NE (and in the
basic version of the folk theorem).  Specifically, players are required
to punish a  
deviating player, even though that might hurt their payoff. Thus,
if a deviation occurs, it might not be the best response of
the players to follow their strategy and punish; thus, such
a punishment is actually an empty threat.  

To deal with this (well known) problem, a number of refinements of NE
have been considered.  The one typically used in dynamic games of perfect
information is \emph{subgame-perfect equilibrium}, suggested by
Selten~\citeyear{Selten65}. A
strategy profile is a subgame-perfect equilibrium if it is a NE at every
subgame of the original game. Informally, this means that at any history
of the game (even those that are not on any equilibrium path), if all
the players follow their strategy from that point on, then no player has
an incentive to deviate. In the context of repeated games 
where players' moves are observed (so that it is a game of perfect
information), 
the folk theorem continues to hold even if the solution
concept used is subgame-perfect
equilibrium~\cite{AS94,fudenberg1986folk,Ru79}. 

We define a computational analogue of subgame-perfect equilibrium that we call
\emph{computational subgame-perfect $\eps$-equilibrium}, where the
strategies involved are polynomial time, and 
deviating players are again restricted to using polynomial-time
strategies. 
There are a number of subtleties that arise in defining this notion.
While we assume that all actions in the underlying repeated game are
observable,  
we allow our TMs to also have memory, which means the action of A TM does
not depend only on the public history. 
Like subgame-perfect equilibrium, our computational solution concept
is intended to capture the intuition that the strategies are in
equilibrium after any possible 
deviation. This means that in a computational subgame-perfect  
equilibrium, at each history for player $i$, player $i$ must make a
(possibly approximate) best response, no matter what his and the other
players' memory states are. 
To compute a computational subgame-perfect $\epsilon$-equilibrium, 
we use the same basic strategy as for NE,
but, as often done to get a subgame-perfect equilibrium (for example
see~\cite{fudenberg1986folk}), we limit the punishment phase length,
so that the players are not incentivized not to punish deviations. 
However, to prove our result, we need to overcome one more significant hurdle.
When using cryptographic protocols, it is often the case 
(and, specifically is the case in the protocol used for NE) that player $i$
chooses a secret (e.g., a secret key for a public-key encryption scheme)
as the result of some 
randomization, and then releases some public information which is a
function of the secret (e.g., a public key).  After 
that public information has been released, another party $j$ typically has a profitable
deviation by switching to the TM $M$ that can break the 
protocol---for every valid public information, there always \emph{exists} some TM $M$ that has the secret ``hardwired'' into it (although there may not be an efficient way of finding $M$ given the information).
We deal with this problem by doing what is often done in practice:
we do not use any key for too long, so that $j$ cannot gain too much by
knowing any one key.

A second challenge we face is that in order to prove that our new proposed
strategies are even an $\epsilon$-NE, we need to show that the
payoff of the \emph{best response} to this strategy is not much
greater than that of playing the strategy. However, since 
for any polynomial-time TM there
is always a better polynomial-time TM that has just a slightly longer
running time, this natural approach fails. This instead leads us to
characterize a class of TMs we can analyze, and show that any other TM
can be converted to a TM in this class that has at least the same
payoff. 
While such an argument might seem simple in the traditional setting,
since we only allow for polynomial time TMs, in our setting this turns
out to require a 
surprisingly delicate construction and analysis to make sure this
converted TM does indeed has the correct size and running time.

The idea of considering resource-bounded agents has a long history in
game theory.  It is known, for example, that cooperation is a NE of
finitely-repeated prisoner's dilemma with resource-bounded players (see,
e.g., \cite{neyman1985,rubinstein1986finite,PY94}).
The idea of using the structure of the game as a means of correlation is
used by Lehrer \citeyear{lehrer1991internal} to show an equivalence
between NE and 
correlated equilbrium in certain repeated games with nonstandard
information structures. 
The use of cryptography in game theory goes back to Urbano and Vila
\citeyear{urbanoVilla2002,urbano2004unmediated}, who also used it to do
coordination between players.  More recently, it has been 
used by, for example, Dodis, Halevi, and Rabin
\citeyear{dodis2000cryptographic}.  

The application of cryptography perhaps most closely related to ours is
by Gossner \citeyear{gossner1998repeated}, who uses cryptographic
techniques to show how any payoff profile that is 
above the players' \emph{correlated} minimax value can be achieved in a
NE of a repeated game with public communication played by computationally
bounded players. In \cite{gossner2000sharing}, a strategy similar to
the one that we use is used to prove that, even without
communication, the same result holds. 
Gossner's results apply only to infinitely-repeated games with 3 players
and no discounting; he claims that his results do not hold for games
with discounting.  
Gossner does not discuss the complexity of finding a strategy of
the type that he shows exists.

Recently, Andersen and Conitzer \citeyear{fast2013Andersen} described an
algorithm for finding NE in repeated games with more than two players
with high probability in
\emph{uniform games}. However, this algorithm is
not guaranteed to work for all games, and uses the limit of means as its
payoff criterion, and not discounting. 

There are a few recent papers that investigate solution concepts for
extensive-form games involving computationally bounded
player~\cite{KN08a,GLR13,HP13}; some of these focus on
cryptographic protocols~\cite{KN08a,GLR13}. 
Kol and Naor~\citeyear{KN08a} discuss refinements of NE in the 
context of cryptographic protocols, but their solution concept 
requires only that on each history on the equilibrium path, the
strategies from that point on form 
a NE. 
Our requirement for the computational subgame-perfect equilibrium is much stronger.
Gradwohl, Livne and Rosen~\citeyear{GLR13} also consider this
scenario and offer a solution concept different from ours; they try to
define when an empty threat occurs, and look for strategy profiles where
no empty threats are made. Again, our solution concept is much
stronger.  

The rest of this paper is organized as follows.  In
Section~\ref{sec:preliminaries}, we review the relevant definitions from  
game theory and cryptography.
In section~\ref{sec:NE}, we define our notion of computational $\epsilon$-NE and show how find it efficiently for repeated games.
In Section~\ref{sec:perfecteq}, we consider
computational subgame-perfect $\epsilon$-equilibrium 
and show that it too can be found efficiently.

\section{Preliminaries}\label{sec:preliminaries}
\subsection{One-shot games}
\newcommand{\mm}{\mathit{mm}}
We define a game $G$ to be a triple $([c],A,\vec{u})$, where
$[c] = \{1,\ldots, c\}$ is the set of players, $A_i$ is the
set of possible actions for player $i$, 
$A=A_1\times\ldots\times A_c$
is the set of action profiles, and 
$\vec{u}:A \to \mathbb{R}^c$ is the utility function
($\vec{u}_i(\vec{a})$ is the utility of player
$i$). A (mixed) \emph{strategy} $\sigma_i$ for player $i$ is a
probability distribution over $A_i$, that is, an element of 
$\Delta(A_i)$ (where, as usual, we denote by $\Delta(X)$ the set of
probability distributions over the set $X$). 
We use the standard notation $\vec{x}_{-i}$ to denote vector $\vec{x}$ with
its $i$th element removed, and $(x',\vec{x}_{-i})$ to denote
$\vec{x}$ with its $i$th element replaced by $x'$. 

\begin{definition} (Nash Equilibrium)
$\sigma=(\sigma_1,...,\sigma_c)$ is an $\epsilon$-NE of $G$ if, for all
players $i\in [c]$ and all actions $a_i'\in A_i$, 
$E_{\sigma_{-i}}[u_i(a_i',\vec{a}_{-i})]\leq  
E_{\sigma}[u_i(\vec{a})]+\epsilon. $
\end{definition}

A \emph{correlated strategy} of a game $G$ is an element
\mbox{$\sigma\in\Delta(A)$}. It is a \emph{correlated equilibrium}
if, for all players $i$, they have no temptation to play a different
action, given that the action profile was chosen according to $\sigma$.
That is, for all players $i$ for all $a_i \in A_i$ such that $\sigma_i(a_i)
> 0$, $E_{\sigma \mid a_i}u_i(a_i,\vec{a}_{-i}) \ge
E_{\sigma \mid a_i}u_i(a_i',\vec{a}_{-i})$.

Player $i$'s minimax value in a game $G$ is
the highest payoff $i$ can guarantee himself if the other
players are trying to push his payoff as low as possible. We call the
strategy $i$ plays in this case a minimax strategy for $i$; the
strategy that the other players use is $i$'s (correlated) punishment
strategy.  
(Of course, there could be more than one minimax strategy or punishment
strategy for player $i$.)  
Note that a correlated punishment strategy can be computed using linear programming.

\begin{definition} 
Given a game $G=([c],A,\vec{u})$, 
the strategies $\vec{\sigma}_{-i} \in \Delta(A_{-i})$ that minimize
$\max_{\sigma'\in\Delta{(A_i)}}E_{(\sigma',\vec{\sigma}_{-i})}[u_i(\vec{a})]$  
are the \emph{punishment strategies} against player $i$ in $G$. 
If $\vec{\sigma}_{-i}$ is a punishment strategy against player $i$, then
$\mm_i(G)=\max_{a\in A_i}E_{\vec{\sigma}_{-i}}[u_i(a,a_{-i})]$ is player
$i$'s \emph{minimax value} in $G$ 
\end{definition}

To simplify the presentation, we assume all payoffs are normalized so that each player's minimax value is 0.
Since, in an equilibrium, all players get at least their minimax value,
this guarantees that all players get at least 0 in a NE.

\subsection{Infinitely repeated games}

Given a normal-form game $G$, we define the repeated game
$G^{t}(\delta)$ as the game in which $G$ is played repeatedly $t$ times
(in this context, $G$ is called the \emph{stage game})
and \mbox{$1-\delta$} is the discount factor (see below). Let
$G^{\infty}(\delta)$ be the game where $G$ is played infinitely many
times. 
An infinite history $h$ in this game is an infinite sequence
$\langle\vec{a}^0,\vec{a}^1, \ldots\rangle$ of action profiles.
Intuitively, we can think of $\vec{a}^t$ as the action profile played in
the $t^{\mathrm{th}}$ stage game.  
We often omit the $\delta$ in $G^{\infty}(\delta)$
if it is not relevant to the discussion.
Let $H_{G^{\infty}}$ be the set of all possible histories of $G^{\infty}$.
For a history $h\in H_{G^{\infty}}$ let $G^{\infty}(h)$ the subgame that
starts at history $h$ (after $|h|$ one-shot games have been played where
all players played according to $h$). 
We assume that $G^\infty$ is \emph{fully observable}, in the
sense that, after each stage game, the players observe exactly what
actions the other players played.

A (behavioral) strategy for player $i$ in a repeated game is a function $\sigma$ from histories of the games to $\Delta(A_i)$. Note that a profile $\vec{\sigma}$ induces a distribution $\rho_{\vec{\sigma}}$
on infinite histories of play.
Let $\rho_{\vec{\sigma}}^t$ denote the induced distribution on  
$H^t$, the set of histories of length $t$.
(If $t=0$, we take $H^0$ to consist of 
the unique history of length 0, namely $\langle \, \rangle$.)
Player $i$'s utility if $\vec{\sigma}$ is played, denoted $p_i(\vec{\sigma})$, is
defined as follows:
\begin{equation*}
p_i(\vec{\sigma})=\delta\sum_{t=0}^\infty (1-\delta)^t
\sum_{h \in H^{t}, \vec{a} \in A} \rho_{\vec{\sigma}}^{t+1}(h \cdot \vec{a})
[u_i(\vec{a})].  
\end{equation*}
Thus, the discount factor is $1-\delta$.
Note that the initial $\delta$ is a normalization factor.  It guarantees
that if  $u_i(\vec{a}) \in [b_1,b_2]$ for all joint actions $\vec{a}$ in
$G$, then $i$'s utility is in $[b_1,b_2]$, no matter which strategy profile
$\vec{\sigma}$ is played.

In these game, a more robust solution concept is subgame-perfect
equilibrium~\cite{Selten65},
which requires that the strategies
form an $\epsilon$-NE at every history of the game.

\begin{definition} 
A strategy profile $\vec{\sigma}=(\sigma_1,...,\sigma_c)$, is
a subgame-perfect $\epsilon$-equilibrium of a repeated game
$G^{\infty}$, if, for all players   
$i\in [c]$, all histories $h\in
H_{G^{\infty}}$ 
where player $i$ moves, and all strategies $\sigma'$ for player $i$,
\begin{equation*}
p_i^h((\sigma')^h,\vec{\sigma}^h_{-i})\leq
p_i^h(\vec{\sigma}^h)+\epsilon, 
\end{equation*}
where $p_i^h$ is the utility function for player $i$ in game
${G}^{\infty}(h)$, and $\sigma^h$ is the restriction of $\sigma$ to
$G^\infty(h)$.   
\end{definition}

\subsection{Cryptographic definitions} 

For a probabilistic algorithm $A$ and an infinite bit string $r$,
$A(x; r)$ denotes the output of $A$ running on input $x$ with randomness
$r$; 
$A(x)$ denotes the distribution on outputs of $A$ induced by considering
$A(x;r)$, where $r$ is chosen uniformly at random. 
A function $\epsilon :  \mathbb{N} \rightarrow [0,1]$ is\emph{ negligible}
if, for every constant $c \in \mathbb{N}$,  $\epsilon(k) < k^{-c}$ for sufficiently large $k$.

We use a \emph{non-uniform} security model, which means our attackers are
\emph{non-uniform} PPT algorithm. 
\begin{definition}
A \emph{non-uniform} probabilistic
polynomial-time machine $A$ is a sequence
of probabilistic machines $A = \{A_1, A_2, . . .\}$ for which there exists a
polynomial $d$ such that both $|A_n|$, the \emph{description size of
  $A_n$}  (i.e., the states and transitions in $A_n$), and the 
running time of $A_n$ are less than $d(i)$. 
\end{definition}

Alternatively, a non-uniform PPT machine can also be defined as a uniform PPT machine that receives an advice string (for example, on an extra ``advice" tape) for each input length.
It is common to assume that the cryptographic building blocks we define next and use in our
constructions are secure against non-uniform PPT algorithms. 

\subsubsection{Computational Indistinguishability}

\begin{definition} A \emph{probability ensemble} is a
sequence $X=\{X_n\}_{n\in\mathbb{N}}$  of probability distribution indexed by
$\mathbb{N}$. 
(Typically, in an ensemble $X=\{X_n\}_{n\in\mathbb{N}}$, the support of $X_n$
consists of strings of length $n$.) 
\end{definition}

We now recall the definition of computational indistinguishability
\cite{goldwasser1984probabilistic}.

\begin{definition} Two probability
ensembles 
$\{X_n\}_{n\in\mathbb{N}}, \{Y_n\}_{n\in\mathbb{N}}$ are \emph{computationally indistinguishable} if, for
all non-uniform
PPT TMs $D$, there
exists a negligible function $\epsilon$ such that, for all $n \in
\mathbb{N}$,  
$$|\Pr[D(1^n, X_n) = 1] - \Pr[D(1^n, Y_n) = 1]| \leq \epsilon(n).$$
To explain the $\Pr$ in the last line, recall that $X_n$ and $Y_n$
are probability distributions. Although we write  
$D(1^n, X_n)$, $D$ is a randomized
algorithm, so what $D(1^n, X_n)$ returns depends on the outcome
of random coin tosses.  To be a little more formal, we should write
$D(1^n, X_n,r)$, where $r$ is an infinitely long random bit strong (of
which $D$ will only use a finite initial prefix).  More
formally, taking $\Pr_{X_n}$ to be the joint distribution over strings
$(x,r)$ where $x$ is chosen according to $X_n$ and r is chosen
according to the uniform distribution on bit-strings,  
we want  
$$|{\Pr}_{X_n}\left[\{(x,r):D(1^n, x,r) =
1\}\right]-{\Pr}_{Y_n}\left[\{(y,r):D(1^n, 
y,r) = 1\}\right]|\leq \epsilon(n).$$   
We similarly abuse notation elsewhere in writing $\Pr$.
\end{definition}

We often call a TM that is
supposed to distinguish between two probability ensembles a
\emph{distinguisher}.

\subsubsection{Pseudorandom Functions}
\begin{definition} A \emph{function ensemble} is a
sequence $F=\{F_n\}_{n\in\mathbb{N}}$ of probability distributions such that the
support of $F_n$ is a set of functions mapping
$n$-bit strings to $n$-bit strings.  
The \emph{uniform function ensemble}, denoted $H=\{H_n\}_{n\in\mathbb{N}}$,
has $H_n$ be the uniformly distribution over the set of all functions mapping
$n$-bit strings to $n$-bit strings. 
\end{definition}

We have the same notion of computational indistinguishablity for
function ensembles as we had for probability ensembles, only that the
distinguisher is now an oracle machine, meaning that it can query the
value of the function at any point with one computation step, although it
does not have the full description of the function. 
(See \cite{goldreichFound} for a detailed description.)

We now define \emph{pseudorandom functions} (see
\cite{goldreich1986construct}). Intuitively, 
this is a family of functions indexed by a seed, such that it is hard
to distinguish a random member of the family from a truly randomly
selected function.

\begin{definition}
A \emph{pseudorandom function ensemble (PRF)} is a set 
\mbox{$\{f_s:\{0,1\}^{|s|}\to\{0,1\}^{|s|}\}_{s\in\{0,1\}^*}$} such that the
following conditions hold: 
\begin{itemize}
\item (easy to compute) $f_s(x)$ can be computed by a PPT algorithm that
is given $s$ and $x$; 

\item (pseudorandom) the function ensemble
\mbox{$F=\{F_n\}_{n\in\mathbb{N}}$}, where $F_n$ is uniformly
distributed over the multiset $\{f_s\}_{s\in\{0,1\}^n}$, is
computationally indistinguishable from $H$. 
\end{itemize}
\end{definition}

We use the standard cryptographic assumption that 
a family of PRFs exists; 
this assumption is implied by the existence of one-way functions
\cite{haastad1999pseudorandom,goldreich1986construct}. 
We actually require the use of a seemingly stronger notion of a PRF,
which requires that an attacker getting access to 
polynomially many instances of a PRF (i.e., $f_s$ for polynomially many
values of $s$) still cannot distinguish them from polynomially many
truly random functions.
Nevertheless, as we show in Appendix~\ref{section:appA}, it follows using a standard
``hybrid'' argument that any PRF satisfies also this stronger
``multi-instance'' security notion.

\subsubsection{Public-key Encryption Schemes}
We now define public-key encryption schemes. 
Such a scheme has two keys. The first is public and used for 
encrypting messages (using a randomized algorithm).
The second is secret and used for decrypting. The keys are generated
in such a way that the probability that a decrypted message is equal
to the encrypted message is equal to $1$. 
The key generation algorithm takes as input a ``security parameter''
$k$ that is used to determine the security of the protocols
(inuitively, no polynomial-time attacker should be able to ``break'' the
security of the protocol except possibly with a probability that is a
negligible function of $k$).

We now recall the formal definitions of public-key encryption schemes
\cite{diffie1976new,rivest1978method,goldwasser1984probabilistic}.
\begin{definition} Given a polynomial $l$,
an \emph{$l$-bit public-key encryption scheme} is a triple \mbox{$\Pi = (\Gen, \Enc, \Dec)$}
of $\PPT$ algorithms where (a) $\Gen$ takes a security parameter
$1^k$ as input and returns  a (public key, private
 key) pair; (b) $\Enc$ takes a public key $pk$ 
 and a message $m$ in a message space $\{0,1\}^{l(k)}$ as input and
returns a ciphertext $\Enc_{pk}(m)$; (c) $\Dec$ is   
a deterministic  algorithm that takes a secret key $sk$ and a
ciphertext $\mathcal{C}$ as input and outputs $m' = \Dec_{sk}(\mathcal{C})$, and (d)
$$\Pr \left[\exists m \in \{0,1\}^{l(k)} \mbox{ such that } \Dec_{sk}(Enc_{pk}(m))\neq m  \right ] = 0 .$$  
\end{definition}

We next define a security notion for public-key encryption. Such a
security 
notion
considers an adversary that is characterized by two PPT algorithms, $A_1$ and
$A_2$. Intuitively, $A_1$ gets as input a public key that is part of a
(public key, secret key) pair randomly generated by \Gen, 
together with a security parameter $k$.
$A_1$ then outputs two messages in $\{0,1\}^k$ (intuitively, messages
it can distinguish), and some side information that it passes to $A_2$
(intuitively, this is information that $A_2$ needs, such as the messages
chosen; An example of how this is used can be seen in
Appendix~\ref{section:appB}).    
$A_2$ gets as input
the encryption of one of those messages
and the side information passed on by $A_1$.
$A_2$ must output which of the
two messages 
$m_0$ and $m_1$ the encrypted message
is the encryption of  
(where an output of $b \in \{0,1\}$ indicates that it is $m_b$).
Since $A_1$ and $A_2$ are PPT algorithms, 
the output of $A_2$ can be viewed as a
probability distribution over $\{0,1\}$. The scheme is secure if the two
ensembles (i.e., the one generated by this process where the encryption
of $m_0$ is always given to $A_2$, and the one where the encryption of
$m_1$ is always given to $A_2$) 
are indistinguishable. More formally: 

\begin{definition}[Public-key security] An $l$-bit
public-key encryption scheme $\Pi = (\Gen, \Enc, \Dec)$ is \emph{secure}
if, for 
every probabilistic polynomial-time adversary $A = (A_1,A_2)$, the
ensembles 
$\{\text{IND}_0^{\Pi}(A, k)\}_k$ and $\{\text{IND}_1^{\Pi}(A,k)\}_k$ 
are computationally indistinguishable, where 
$\{\text{IND}_b^{\Pi}(A,k)\}_k$ is the following PPT algorithm:
\begin{center}
\begin{tabular}[c]{rlcrl}

$\text{IND}_b^{\Pi}(A, k) :=$ & $(pk,sk) \leftarrow \Gen(1^k)$ \\
 & $(m_0, m_1, \tau) \leftarrow A_1(1^k, pk)$ $(m_0,m_1\in\{0,1\}^k)$\\
 & $\mathcal{C} \leftarrow \Enc_{pk}(m_b)$ \\
 & $o \leftarrow A_2(\mathcal{C}, \tau)$ \\
 & $\text{Output } o$.%
\end{tabular}
\end{center}

\noindent Intuitively, the $\leftarrow$ above functions as an assignment
statement, but it is not quite that, since the various algorithms are
actually PPT algorithms, so their output is randomized.  
Formally, $\text{IND}_b^{\Pi}(A, k)$ is a 
probability distribution, which we can write as $\text{IND}_b^{\Pi}(A, k, r_1,
r_2,r_3,r_4)$,  where we view $r_1$, $r_2$, $r_3$, and $r_4$ as the random
bitstrings that 
serve as the
second arguments of $\Gen$, $A_1$, $\Enc_{pk}$, and $A_2$,
respectively.  Once we add these arguments 
(considering, e.g.,  $\Gen(1^k, r_1)$ and $A_1(1^k, pk,r_2)$ rather than  
$\Gen(1^k)$ and $A_1(1^k, pk)$)
these algorithms become deterministic, and $\leftarrow$ can indeed be
viewed as an assignment statement.
\end{definition}
 
We assume a secure public-key encryption scheme exists. 
We actually require a seemingly stronger notion of
``multi-instance'' security, where an attacker gets to see encryptions
of multiple messages, each of which is encrypted using multiple keys.

\begin{definition}
An $l$-bit public-key encryption scheme $\Pi = (\Gen, \Enc, \Dec)$ is 
\emph{multi-message multi-key secure} if, for all polynomials $f$ and
$g$, and for every probabilistic polynomial
time adversary \mbox{$A = (A_1,A_2)$}, the ensembles
$\{\text{IND-MULT}_0^{\Pi}(A, k,f,g)\}_k$ and $\{\text{IND-MULT}_1^{\Pi}(A,
k,f,g)\}_k$ are computationally indistinguishable, where 
\begin{center}
\noindent\begin{tabular}[c]{rlcrl}
& $\text{IND-MULT}_b^{\Pi}(A, k,f,g) :=$\\
 & $(pk_1,sk_1)
 \leftarrow \Gen(1^k),\ldots (pk_{g(k)},sk_{g(k)}) \leftarrow \Gen(1^k),$ \\ 
 & $(m_0^1,\ldots,m_0^{f(k)}, m_1^1,\ldots,m_1^{f(k)},\tau) \leftarrow A_1(1^k, pk_1,\ldots,pk_{g(k)})$ 
($m_0^i,m_1^i\in\{0,1\}^k$)\\
 & $\mathcal{C} \leftarrow \Enc_{pk_1}(m_b^1),\ldots,
 \Enc_{pk_{g(k)}}(m_b^1),\ldots,\Enc_{pk_1}(m_b^{f(k)}),\ldots,
 \Enc_{pk_{g(k)}}(m_b^{f(k)})$ \\
& $o \leftarrow A_2(\mathcal{C},\tau)$ \\
 & $\text{Output } o$ \\

\end{tabular}
\end{center}
\end{definition}

In this definition, there are polynomially many messages
being encrypted, and each message is encrypted 
a polynomial number of times, using a different key each time.
Other than that, the process is similar
to the standard definition of security.
As we show in Appendix~\ref{section:appB}, any secure encryption scheme is also
multi-message multi-key secure.

\section{The complexity of finding $\epsilon$-NE in repeated games played by stateful machines}\label{sec:NE}
\newcommand{\sq}{\mathit{sq}}
\newcommand{\seed}{\mathit{seed}}
\newcommand{\NE}{\mathit{NE}}
\newcommand{\G}{\mathcal{G}}
\newcommand{\PS}{\mathit{PS}}

\subsection{Computational NE Definition}

Since we consider computationally-bounded players, we take a player's
strategy in $G^{\infty}$ to be  a (possibly
probabilistic) Turing machine (TM), which outputs at each round an action to
be played, based on its internal memory and the history of play so far.
(The TMs considered in BC+ did not have internal memory.)
We consider only TMs that at round $t$ use polynomial in $nt$ many steps
to compute the next action, 
where $n$ is the maximum number of actions a player has in $G$.  
Thus, $n$ is a measure of the size of $G$.%
\footnote{When we talk about polynomial-time algorithms, we mean
polynomial in $n$.  We could use other measures of the size of $G$, such
as the total number of actions.  Since all reasonable choices of size
are polynomially related, the choice does not affect our results.}
Denote by $M_i$ the TM used by player $i$, and let $\vec{M}=(M_1,\ldots ,M_c)$.

Note that a profile $\vec{M}$ induces a distribution $\rho_{\vec{M}}$
on infinite histories of play.
Let $\rho_{\vec{M}}^t$ denote the induced distribution on  
$H^t$, the set of histories of length $t$.
(If $t=0$, we take $H^0$ to consist of 
the unique history of length 0, namely $\langle \, \rangle$.)
Player $i$'s utility if $\vec{M}$ is played, denoted $p_i(\vec{M})$, is
defined as follows:
\begin{equation*}
p_i(\vec{M})=\delta\sum_{t=0}^\infty (1-\delta)^t
\sum_{h \in H^{t}, \vec{a} \in A} \rho_{\vec{M}}^{t+1}(h \cdot \vec{a})
[u_i(\vec{a})].  
\end{equation*}
Thus, the discount factor is $1-\delta$.
Note that the initial $\delta$ is a normalization factor.  It guarantees
that if  $u_i(\vec{a}) \in [b_1,b_2]$ for all joint actions $\vec{a}$ in
$G$, then $i$'s utility is in $[b_1,b_2]$, no matter which TM profile
$\vec{M}$ is played.

We are now ready to define the notion of equilibrium we
use. Intuitively, as we model players as polynomial-time TMs, we
consider a profile of TMs an equilibrium in a game if there is no player
and no other polynomial-time TM 
that gives that player a higher expected payoff (or up to an $\epsilon$ for an $\epsilon$-NE). 

Since we consider (probabilistic) TMs that run in polynomial time in the
size of the game, we cannot consider a single game.  For
any fixed game, 
running in polynomial time in the size of the game
is meaningless.  Instead, we
need to consider a sequence of games. This leads to the following
definition. 

\begin{definition} 
An infinite sequence of strategy profiles
$\vec{M}^1,\vec{M}^{2},\ldots$, where $\vec{M}^k=(M^k_1,...,M^k_c)$ is
an $\epsilon$-NE of an infinite sequence of games $G_1^{\infty},G_{2}^{\infty},\ldots$
where the size of $G_k$ is $k$ if, for all players 
$i\in [c]$ and all non-uniform PPT adversaries 
$\bar{M}$ (polynomial in $k$ and 
$t$, as discussed above), 
there exist $k_0$ such that for all $k\geq k_0$
\begin{equation*}
p^k_i(\bar{M},\vec{M}^k_{-i})\leq p^k_i(\vec{M}^k)+\epsilon(k).
\end{equation*}
where $p^k_i$ is the payoff of player $i$ in game $G_k^{\infty}$.
\end{definition}

We note that the equilibrium definition we use considers only deviations
that can be implemented by non-uniform polynomial-time TMs. This is different from
both the usual definition of NE and from the definition used by BC+,
who allow arbitrary deviations. 
But this difference is exactly what allows us to use cryptographic techniques.
The need to define
polynomial-time deviation is the reason for considering sequences of
games instead of a single game. 
There are other reasonable ways of capturing
polynomial-time adversaries. As will be seen from our proof, our
approach is quite robust, so our results should hold for any reasonable
definition. 

\subsection{Computing an equilibrium}

In this section we describe the equilibrium strategy and show how to efficiently compute it. We first start with some definition and tool we need for our proof.

\subsubsection{Preliminaries}
\begin{definition}
Let $\G_{a,b,c,n}$ be the set of all games with $c$ players, at most $n$
actions per player, 
integral payoffs\footnote{Our result also hold for rational payoffs 
except then the size of the game needs to take into account the
bits needed to represent the payoffs}, 
maximum payoff $a$, and minimum payoff $b$. 
\end{definition}

Note that by our assumption that the minimax payoff is 0 for all
players, we can 
assume $a\ge0$, $b\le0$, and $a-b > 0$ (otherwise $a=b=0$, which makes
the game uninteresting).
We start by showing that, given a correlated strategy $\sigma$ in a game $G$,
players can get an average payoff that is arbitrarily close to their
payoff in $\sigma$ by playing a fixed sequence of action profiles repeatedly.
\begin{lemma}\label{sequence}
For all $a$, $b$, $c$,  all polynomials $q$, 
all $n$, all games 
$G\in \G_{a,b,c,n}$, and all correlated strategies $\sigma$ in $G$, 
if the expected payoff vector of playing $\sigma$ is $p$
then there
exists a 
sequence $\sq$ of length $w(n)$,
where $w(n) = ((a-b)q(n) +1)n^c$,
such that player
i's average payoff in $\sq$ is at least $p_i-1/q(n)$. 
\end{lemma}
\begin{proof}
Given $\sigma$, we create $\sq$ the obvious way: by playing each action
in proportion to the probability $\sigma(\vec{a})$.  
More precisely, let $r =a-b$, and define
\mbox{$w(n)=(r q(n)+1) n^c$}, as in the statement of the lemma.  We create a
sequence $\sq$ by 
playing each action profile $\vec{a}$ $\lfloor w(n)\sigma(\vec{a})
\rfloor$ times, in some fixed order. Notice that the length of this
sequence is between $w(n) - n^c$ and $w(n)$. The average payoff player
$i$ gets  
in $\sq$ is
\begin{align*}
v_i'&=\frac{1}{\sum_{\vec{a}\in A}\lfloor w(n)\sigma(\vec{a}) \rfloor}\sum_{\vec{a}\in A}\lfloor w(n)\sigma(\vec{a}) \rfloor u_i(\vec{a})\\
& \geq \frac{1}{\sum_{\vec{a}\in A}\lfloor w(n)\sigma(\vec{a}) \rfloor}
\left(\sum_{\vec{a}\in A, u_i(\vec{a})\geq 0} (w(n)\sigma(\vec{a})-1)
u_i(\vec{a})+\sum_{\vec{a}\in A, u_i(\vec{a})<0} w(n)\sigma(\vec{a})
u_i(\vec{a})\right) \\ 
& = \frac{w(n)\sum_{\vec{a}\in A}\sigma(\vec{a})
u_i(\vec{a})}{\sum_{\vec{a}\in A} \lfloor w(n)\sigma(\vec{a})\rfloor}
-\frac{\sum_{\vec{a}\in A, u_i(\vec{a})\geq
0}u_i(\vec{a})}{\sum_{\vec{a}\in A}\lfloor w(n)\sigma(\vec{a}) \rfloor}
\geq  \frac{w(n)p_i}{\sum_{\vec{a}\in A}\lfloor w(n)\sigma(\vec{a})
\rfloor} -\frac{a n^c}{ w(n)-n^c}.
\end{align*}
If $p_i< 0$,
\begin{align*}
v_i'&\geq \frac{w(n)p_i}{\sum_{\vec{a}\in A}\lfloor w(n)\sigma(\vec{a}) \rfloor} -\frac{a n^c}{ w(n)-n^c}  \geq \frac{w(n)p_i-an^c}{w(n)-n^c}\\
&= \frac{(rq(n)+1)n^c p_i-a n^c}{(rq(n)+1)n^c-n^c}=  \frac{rq(n)n^c
p_i-(a-p_i)n^c}{rq(n)n^c}\geq p_i-\frac{1}{q(n)}. 
\end{align*}
If $p_i\geq 0$,
\begin{align*}
v_i'&\geq \frac{w(n)p_i}{\sum_{\vec{a}\in A}\lfloor w(n)\sigma(\vec{a}) \rfloor} -\frac{a n^c}{ w(n)-n^c}  \geq p_i-\frac{an^c}{w(n)-n^c}\\
&=p_i -\frac{a n^c}{(rq(n)+1)n^c-n^c}=p_i-  \frac{an^c}{rq(n)n^c}\geq p_i-\frac{1}{q(n)} .
\end{align*}
\end{proof}

\begin{lemma}\label{derandomize}
For all $a$, $b$, $c$, all polynomials $q$  and $w$, 
all $G\in \G_{a,b,c,n}$, and all sequences $\sq$ of length $w(n)$, 
if the average payoff vector of playing $\sq$ is $p$,
 then   
for all $\delta \le 1/f(n)$, 
where $f(n) = (a-b) w(n) q(n)$,
 if $sq$ is
played infinitely often, player i's payoff in $G^{\infty}(\delta)$ 
is at least
$p_i-1/q(n)$. 
\end{lemma}

\begin{proof}
Suppose that $\sq=(a_0,\ldots,a_{w(n)-1})$, and 
let $v_i$ be $i$'s payoff from $\sq^\infty$ in $G^\infty(\delta$).  Then
\begin{align*}
v_i&=\delta\sum_{t=0}^{\infty}(1-\delta)^{t w(n)}\sum_{k=0}^{w(n)-1}u(a_k)(1-\delta)^k \\
&=p_i+\delta\sum_{t=0}^{\infty}(1-\delta)^{t w(n)}\sum_{k=0}^{w(n)-1}(u(a_k)-p_i)(1-\delta)^k. 
\end{align*}
We want to bound the loss from the second part of the sum. Notice that
this is a discounted sum of a sequence whose average payoff is $0$. Call
this sequence $\sq'$.  Observe that, 
because of the discounting, in the worst case, $i$ gets all of his
negative payoff in the first round of $\sq'$ and all his positive
payoffs in the last round.  Thus, we can bound the discounted average
payoff by analyzing this case.  Let the sum of $i$'s negative payoffs in
$\sq'$ be $P_{neg}$, 
which means that the sum of $i$'s positive payoffs must be
$-P_{neg}$. Let $r=a-b$, let  
$v'_i = \min_{\vec{a} \in A} (u_i(\vec{a})-p_i)\geq -r $, and let  
$f(n)=r w(n) q(n)$, as in the statement of the lemma.  
So, if $\delta \le 1/f(n)$,
 player $i's$ average discounted payoff in the game is at least
\begin{align*}
v_i & \geq p_i+\delta\sum_{t=0}^\infty P_{neg}(1-\delta)^{w(n)t}+(-P_{neg})(1-\delta)^{w(n)(t+1)-1}\\
&=p_i+\delta( P_{neg}+(-P_{neg})(1-\delta)^{w(n)-1})\sum_{t=0}^\infty(1-\delta)^{w(n)t}\\
& =p_i+ \delta( P_{neg}+(-P_{neg})(1-\delta)^{w(n)-1})\frac{1}{1-(1-\delta)^{w(n)}} \\
& =
p_i+P_{neg}\delta\frac{1-(1-\delta)^{w(n)-1}}{(1-(1-\delta)^{w(n)})}\geq p_i+\delta P_{neg}\geq p_i+\frac{P_{neg}}{f(n)}\geq 
p_i+\frac{v_i'w(n)}{f(n)} = p_i-1/q(n). 
\end{align*}
\end{proof}

The next lemma shows that, for every inverse polynomial, if we ``cut off''
the game after some 
appropriately large polynomial $p$ number of rounds (and compute the
discounted utility for the finitely repeated game considering only
$p(n)$ repetitions), each player's utility in the finitely repeated
game is negligibly close to his utility in the infinitely repeated
game---that is, the finitely repeated game is a ``good'' approximation
of the infinitely repeated game.

\begin{lemma}\label{negligible}
For all $a$, $b$, $c$, all polynomials $q$, all $n$, all games
$G\in \G_{a,b,c,n}$, all $0<\delta<1$, all strategy profiles
$\vec{M}$, and all players $i$, $i$'s expected utility $p_i[\vec{M}]$ in game 
$G^{\lceil n/\delta\rceil}(\delta)$ and $p_i[\vec{M}]$ in game
$G^{\infty}(\delta)$ differ by at most $a/e^{n}$.
\end{lemma}

\begin{proof}
Let $p_i^t(\vec{M})$ denote player $i$'s expected utility if the players
are playing 
$\vec{M}$ and the game ends at round $t$. 
Recall that $(1-\delta)^{1/\delta} \le 1/e$.
$$\begin{array}{lll}
&p_i^{\infty}(\vec{M})-p_i^{\lceil n/\delta\rceil}(\vec{M})\\
=&\delta\sum_{t=0}^\infty
(1-\delta)^t 
\sum_{h \in H^{t}, \vec{a} \in A} \rho_{\vec{M}}^{t+1}(h \cdot \vec{a})
[u_i(\vec{a})]-
\delta\sum_{t=0}^{\lceil n/\delta\rceil} (1-\delta)^t
\sum_{h \in H^{t}, \vec{a} \in A} \rho_{\vec{M}}^{t+1}(h \cdot \vec{a})
[u_i(\vec{a})]\\
= &\delta\sum_{t=\lceil n/\delta\rceil+1}^{\infty} (1-\delta)^t
\sum_{h \in H^{t}, \vec{a} \in A} \rho_{\vec{M}}^{t+1}(h \cdot \vec{a})
[u_i(\vec{a})]\\
\leq& \delta\sum_{t=\lceil n/\delta\rceil}^{\infty} (1-\delta)^t a\\
= &\delta(1-\delta)^{\lceil n/\delta\rceil}\sum_{t=0}^{\infty}(1-\delta)^ta
=\delta(1-\delta)^{\lceil n/\delta\rceil}\frac{a}{\delta}\leq \frac{a}{e^n}.
\end{array}
$$ 
\end{proof}
\subsubsection{The $\epsilon$-NE strategy and the algorithm}\label{sec:computingNE}

Let \mbox{$A_i^0\subset A_i$} be a non-empty set and let \mbox{$A_i^1=A_i \setminus A_i^0$}.\footnote{We assume that each player has at least two actions in $G$. 
This assumption is without loss of generality---we can essentially
ignore players for whom it does not hold.}
A player can broadcast an $m$-bit string by using his actions for $m$ rounds, by treating actions from $A_i^0$ as 0 and actions from $A_i^1$ as 1.
Given a polynomial $\phi$ (with natural coefficients),
let $(\Gen,\Enc,\Dec)$ be a multi-message multi-key secure $\phi$-bit,
if the security parameter is $k$, the length of an encrypted message is
$z(k)$ for some polynomial $z$. 
Let $\sq=(s_1,s_2\ldots ,s_m)$ be a fixed sequence of action profiles.    
Fix a polynomial-time pseudorandom function ensemble $\{\PS_s: s \in
\{0,1\}^*\}$. 
For a game $G$ such that $|G|=n$, consider the   
strategy $\sigma^{\NE}$ for player $i$ in $G^{\infty}(\delta)$ that has the following three phases. Phase
1 explains what to do if no deviation occurs: play $\sq$.  Phase 2
gives the preliminaries of what to do if a deviation does occur:
roughly, compute a random seed that is shared with all the
non-deviating players.   Phase 3 explains how to use the random seed to
produce a correlated punishment strategy that punishes the deviating
player. Formally let $\vec{M}^{\sigma^{\NE}}$ be the TMs that implement the following strategy:

\begin{enumerate}
\item Play according to $\sq$ (with wraparound) as long as
all players played according to $\sq$ in the previous round.
\item After detecting a deviation by player $j \ne i$
in round $t_0$:%
\footnote{If more than one player deviates while playing $\sq$, the
players punish the one with the smaller index. The punished player
plays his best response to what the other players are doing in this phase.} 
	\begin{enumerate}
\item Generate a pair $(pk_i,sk_i)$ using
$\Gen(1^n)$.
Store $sk_i$ in memory and 
use the next $l(n)$ rounds to broadcast $pk_i$,
as discussed above.
\item If $i = j+1$ (with wraparound), player $i$ does the following: 
\begin{itemize}
\item $i$ records $pk_{j'}$ for all players $j' \notin \{i,j\}$;
\item $i$ generates a random $n$-bit seed $\seed$;
\item for each player $j' \notin \{i,j\}$, $i$ computes
$m=Enc_{pk_{j'}}(\seed)$, 
and uses the  next \mbox{$(c-2)z(n)$} rounds to communicate these strings
to the players other than $i$ and $j$ (in some predefined order). 
\end{itemize}
\item If $i \ne j+1$, player $i$ does the following:
\begin{itemize}
\item $i$ records the actions played by $j+1$ at time slots designated
for $i$ to retrieve $Enc_{Pk_{i}}(\seed)$;
\item $i$ decrypts to obtain $\seed$, using $Dec$ and $sk_i$.
\end{itemize}
\end{enumerate}
\item Phase 2 ends after $ \phi(n)+(c-2)z(n)$ rounds.  The players 
other than $j$ then
compute $\PS_{\seed}(t)$ and use it to determine which action profile to
play according to the distribution defined by   
a fixed (correlated) punishment strategy against $j$. 
\end{enumerate}

Note that if the players other than $j$ had played a punishment strategy
against $j$, then $j$ would get his minimax payoff of 0.  What the
players other than $j$ are
actually doing is 
playing an approximation to a punishment strategy in two senses: first
they are using a psuedorandom function to generate the randomness, which
means that they are not quite playing according to the actual punishment
strategy.  Also, $j$ might be able to guess which pure strategy profile they
are actually playing at each round, and so do better than his minimax
value.  As we now show, $j$'s expected gain during the punishment phase 
is negligible. 

\begin{lemma}\label{psfMinmax}
For all $a$, $b$, $c$, all polynomials $t$ and $f$, all $n$, and all
games $G\in \G_{a,b,c,n}$, in $G^{\infty}(1/f(n))$, if the players
other than $j$ play 
$\vec{M}^{\sigma^{\NE}}_{-j}$, then if $j$ deviates at round $t(n)$, $j$'s 
expected payoff during the punishment phase is negligible.
\end{lemma}
\begin{proof}
Since we want to show $j$'s expected payoff during the punishment phase (phase (3) only) is negligible, it suffices 
to consider only polynomially many rounds of playing phase (3) (more
precisely, at most $nf(n)$ rounds); 
by Lemma \ref{negligible}, any payoff beyond then is guaranteed to be
negligible 
due to the discounting. 

We construct three variants of the strategy $\vec{M}^{\sigma^{\NE}}_{-j}$, that
vary in phases (2) and (3).  We can think of these variants as
interpolating between the strategy above and the use of true randomness.
(These variants assume an oracle that provides appropriate information;
these variants are used only to make the claims precise.)
\begin{description}
\item[H1]
In phase (2), the punishing players send their public keys to 
$j+1$.  For each player $j'$ not being punished, player $j+1$ then
encrypts the seed $0$ using $(j')$'s public key, and then sends the
encrypted key to $j'$.
In phase (3), the punishing players get the output of a truly random
function (from an oracle), and use it to play the true punishment
strategy.  (In this case, phase (2) can be eliminated.)
\item[H2]
In phase (2), the punishing players send their public keys to
$j+1$.  For each player $j'$ not being punished, player $j+1$ 
encrypts the seed $0$ using $(j')$'s public key, and then sends the
encrypted key to $j'$.
In phase (3), the punishing players get a joint random seed $seed$ (from
an oracle) and use the outputs of 
$\PS_{\seed}$ to decide which strategy profile to play in each round. 
(Again, in this case, phase (2) can be eliminated.)
\item[H3]
In phase (2), the punishing players send their public keys to
$j+1$.  Player $j+1$ chooses a random seed $\seed$ and, for each player
$j'$ not being punished, $j+1$  
encrypts $\seed$ using $(j')$'s public key, and then sends the
encrypted key to $j'$.
In phase (3), the punishing players use the outputs of $\PS_{\seed}$ to
decide which strategy profile to play in each round.   
\end{description}

It is obvious that in $H1$, $j$'s expected payoff is
negligible. (Actually, there is a slight subtlety here.  As we observed
above, using linear programming, we 
can compute a strategy that gives the correlated minimax, which gives
$j$ an expected payoff of 0.  To actually implement this correlated
minimax, the players need to sample according to the minimax distribution.
They cannot necessarily do this exactly
(for example, 1/3 can't be computed exactly using random bits). 
 However, given $n$, the 
distribution can be discretized to the closest
rational number of the form $m/2^n$ using at most $n$ random bits.
Using such a discretized distribution, the players other than $j$ can
ensure that $j$ gets only a negligible payoff.)

We now claim that in $H2$, $j$'s expected payoff during the punishment
phase is negligible. Assume for contradiction that a player playing
$H2$ has a non-negligible payoff 
$\mu(n)$ for all $n$ (i.e., there exists some polynomial $g(\cdot)$
such that $\mu(n) \geq 1/g(n)$ for infinitely many $n$.).
Let \mbox{$h(n)=n (a-b)^2 (1/\mu(n))^2$}.
We claim that if $j$'s expected payoff is non-negligible,
then we can distinguish 
$h(n)$ instances of the PRF $\{PS_s: s \in \{0,1\}^n\}$ with
independently generated random 
seeds, from $h(n)$ independent
truly random functions, contradicting the multi-instance security of
the PRF $PS$.

More precisely, we construct a distinguisher $D$ that, given $1^n$ and
oracle access to a set of 
functions $f^1, f^2, \ldots, f^{h(n)}$, proceeds as follows. It
simulates $H_2$ 
(it gets the description of the machines to play as its non-uniform advice)
$h(n)$ times where in iteration $i'$, it uses the
function $f^{i'}$ as the randomization source of the correlated punishment strategy.
$D$ then computes the average payoff of player $j$ in the $h(n)$
runs, and outputs $1$ if this average exceeds
$\mu(n)/2$. 
Note that if the functions $f^1, f^2, \ldots, f^{h(n)}$ are 
truly independent random functions, then $D$ perfectly simulates
$H_1$ and thus, in each iteration $i'$, 
the expected payoff of player $j$ (during the punishment phase) is negligible.
On the other hand, if the functions $f^1, f^2, \ldots, f^{h(n)}$ are 
$h(n)$ independent randomly chosen instances of the PRF $\{PS_s: s \in
\{0,1\}^n\}$, 
then $D$ perfectly simulates $H_2$, and thus,
in each iteration $i'$,
the expected payoff of player $j$ (during the punishment phase) is at
least $\mu(n)$. 

By Hoeffding's inequality \cite{hoeffding1963probability}, given $m$
random variables $X_1,\ldots,X_m$ all of which take on values in an
interval of size $c'$,
$p(|\overline{X}-E(\overline{X})|\geq r)\leq 2\exp(-\frac{2mr^2}{c'^2})$ .  
Since, in this setting, the range of the random variables is 
an interval
of size $a-b$,
the probability that $D$ outputs $1$ when the function are truly independent
is at most $2/e^n$, while the probability that $D$ outputs $1$ when the
functions are independent randomly chosen instances of the PRF $\{PS_s:
s \in \{0,1\}^n\}$  is at least $1-2/e^n$.  This, in turn, means that
the 
difference 
between them is not negligible, which 
is a contradiction.
Thus, $j$'s expected
payoff in $H2$ must be negligible.  
 
We now claim that in $H3$, player $j$'s expected payoff during the
punishment phase is also negligible. Indeed, 
if $j$ can get a non-negligible payoff, then we can
break the multi-message multi-key secure 
encryption scheme. 

Again, assume for contradiction that the punished player's expected
payoff in the punishment phase is 
a non-negligible function $\mu(n)$ for all $n$.
We can build a distinguisher $A=(A_1,A_2)$
(which also gets the description of the machines to play as its
non-uniform advice) 
to distinguish
$\{\text{IND-MULT}_0^{\Pi}(A, n,h,c)\}_n$ and $\{\text{IND-MULT}_1^{\Pi}(A,
n,h,c)\}_n$
(where we abuse notation and identify $c$ with 
the
constant polynomial that
always returns $c$).
Given $n$, 
$A_1$ randomly selects $h(n)$ messages $r_1,\ldots,r_{h(n)}$ and outputs
$(0,\ldots,0,r_1,\ldots,r_{h(n)},(pk_1,\ldots,pk_c))$.  
$A_2$ splits its input into pieces. The first piece contains the first
$c$ encryptions in $\mathcal{C}$
(i.e., the $c$ encryptions of the first message chosen, according to the
$c$ different encryption functions),
the second the next $c$ encryptions and so
on. Notice that each piece consists of $c$ different encryptions of the same
message in both cases.   
It can also simulate phase (1) by following the strategy for $t$
rounds. It then uses each piece, along with the public
keys, to simulate the communication in phase (2).  For piece $j$ it uses
$r_j$ as the seed of the PRF in phase (3).  
It repeats this experiment for all the different pieces of the input,
for a total of $h(n)$ times, and outputs $1$ if the punished player's
average payoff over all experiments using its strategy is more than $\mu(n)/2$.  

Note that if $b=1$, player $j$ faces $H3$ 
(i.e., the distributions over runs when $b=1$ is identical to the
distribution over runs with H3, since in both cases the seed is chosen
at random and the corresponding messages are selected the same way),
so player j's expected payoff in the
punishment phase is 
$\mu(n)$.
Thus, by Hoeffding's inequality
the probability that player j's average payoff in the punishment phase
is more then 
$\mu(n)/2$ is 
$1-2/e^n$, so $A_2$ outputs $1$ with that probability in the case $b=1$. 
On the other hand, if $b=0$, then this is just $H_2$. We know player $j$'s
expected payoff in the punishment phase in each experiment is no more than negligible in $H_2$, so the probability that the
average payoff  is more than 
$\mu(n)/2$
after $h(n)$ rounds, is negligible. 
This means that there is a non-negligible difference between the
probability $A$ outputs $1$ when $b=1$ and when $b=0$,  
which contradicts the assumption that the encryption scheme is  
multi-message multi-key secure public key
secure.  
Thus, the gain in $H3$ must be negligible. 

$H3$ is exactly the game that the punished player faces; thus, this
shows he can't hope to gain more than a negligible payoff in
expectation. 
\end{proof}

We can now state and prove our main theorem, which says that
$\sigma^{\NE}$ is an $\epsilon$-NE for all inverse polynomials $\epsilon$ and the can be computed in polynomial time. 

\begin{theorem}\label{eqExist}
For all $a$, $b$, $c$, and all polynomials $q$, there is a polynomial $f$
and a polynomial-time algorithm $F$
such that, 
for all sequences $G_1, G_2, \ldots$ of games with $G^j\in G_{a,b,c,j}$
and for all inverse polynomials $\delta\leq1/f$,
the sequence of outputs of $F$ given the sequence $G_1, G_2,
\ldots$ of inputs is a $\frac{1}{q}$-NE
for $G_1^{\infty}(\delta(1)), G_2^{\infty}(\delta(2)), \ldots$.
\end{theorem}

\begin{proof}
Given a game $G^n \in \G(a,b,c,n)$,
the first step of the algorithm is to find a correlated equilibrium
$\sigma$ of
$G^n$. This can be done in polynomial time using linear programming. 
Since the minimax value of the game is 0 for all players, all players
have an expected utility of at least 0 using $\sigma$.
Let $r=a-b$. By Lemma \ref{sequence}, we can construct a sequence $\sq$ of
length $w(n)=(3r n q(n)+1)n^c$ that has an 
average payoff for each player that is at most $1/3q(n)$ less than his
payoff using $\sigma$.
By Lemma \ref{derandomize}, it follows that by setting the discount factor $
\delta<1/f'(n)$, where $f'(n)=3r w(n)q(n)$, the loss 
due to discounting is 
also at most $1/3q(n)$. 
We can also find a punishment strategy against each player in polynomial time,
using linear programming. 

We can now compute the strategy $\vec{M}^{\sigma^{\NE}}$
described earlier that uses the sequence $\sq$ and the
punishment strategies.  Let $\vec{\sigma}^*_n$ be this strategy when given $g_n$ as input.
Let $m(n)$ be the
length of phase (2).   (Note that $m(n)$ is a polynomial that depends
only on the choice of encryption scheme---that is, it depends on $l$,
where an $l$-bit public-key encryption scheme is used, and on $z$, where
$z(k)$ is the length of encrypted messages.)  Let
$$f(n)=\max(3q(n)(m(n)a+1),f'(n)).$$ 
Notice that $f$ is independent of the actual game as required.

We now show that $\vec{\sigma}^{*}_{1},\ldots$ as defined above is a $(1/q)$-NE.   
If in game $G^n$ a player follows $\sigma^{*}_n$, he gets at least $-2/3q(n)$. 
Suppose that player $j$ defects at round $t$; that is, that he plays
according to $\sigma^{*}_n$ until round $t$, and then defects. 
By Lemma \ref{negligible} if $t>\frac{n}{\delta(n)}$,
then any gain from defection is negligible, so
there exists some $n_1$ such that, for all $n>n_1$,   
a defection in round $t$ cannot result in the player gaining more than $\frac{1}{q(n)}$.
If player $j$ defects at round $t\leq \frac{n}{\delta(n)}$, he
gets at most $a$ for the duration of phase (2), which is at most $m(n)$
rounds, 
and then, by Lemma \ref{psfMinmax}, gains only a negligible 
amount, say $\epsilon_{neg}(n)$ 
(which may depend on the sequence of deviations), 
in phase (3). 
Let $u_i^n$ be the payoff of player $i$ in game $G^n$ of the sequence.
It suffices to show that 
\begin{align*}
\delta(n)(\sum_{k=0}^t
u^n_i(a_k)(1-\delta(n))^k+\sum_{k=0}^{m(n)}a(1-\delta(n))^{k+t}+(1-\delta(n))^{t+m(n)}\epsilon_{neg}(n))
-1/q(n)
\\ \leq 
\delta(n)(\sum_{k=0}^t
u^n_i(a_k)(1-\delta(n))^k+\sum_{k=t}^{\infty}u^n_i(a_k)(1-\delta(n))^k).
\end{align*}
By deleting the common terms from both side, rearranging,
and noticing that
 \mbox{$(1-\delta(n))^{m(n)}\epsilon_{neg(n)}\leq\epsilon_{neg(n)}$},
it follows that it suffices to show
\begin{align*}
\delta(n)(1-\delta(n))^t(\sum_{k=0}^{m(n)}a(1-\delta(n))^{k}+\epsilon_{neg}(n))-\frac{1}{q(n)} 
\leq
\delta(n)(1-\delta(n))^t(\sum_{k=0}^{\infty}u_i^n(a_{k+t})(1-\delta(n))^k).
\end{align*}
We divide both sides of the equation by $(1-\delta(n))^t$ .
No matter at what step of the sequence the defection happens, the future
expected discounted payoff from that point on is still 
at least $-2/3q(n)$, as our bound applies for the worst sequence for
a player, and we assumed that in equilibrium all players get at least $0$.
It follows that we need to show
\begin{align*}
\delta(n)(\sum_{k=0}^{m(n)}a(1-\delta(n))^k+\epsilon_{neg}(n))-\frac{1}{q(n)(1-\delta(n))^t}
\leq -\frac{2}{3q(n)}. 
\end{align*}
Since $\epsilon_{neg}$ is negligible  
for all deviations, it follows that, for all sequences of deviations, 
there exists $n_0$ such
that $\epsilon_{neg}(n)<1$ for all $n\geq n_0$.  For $n \geq n_0$,
$$\begin{array}{ll}
&\delta(n)(\sum_{k=0}^{m(n)}a(1-\delta(n))^k+\epsilon_{neg}(n))-\frac{1}{q(n)(1-\delta(n))^t}\\ \leq &\delta(n)(m(n)a+\epsilon_{neg}(n)) -\frac{1}{q(n)} \\
\leq &\frac{m(n)a+\epsilon_{neg}(n)}{f(n)} -\frac{1}{q(n)}\\
 \leq &\frac{m(n)a+\epsilon_{neg}(n)}{3q(n)(m(n)a+1)}-\frac{1}{q(n)}\\
 \leq &\frac{1}{3q(n)}-\frac{1}{q(n)}\\= &-\frac{2}{3q(n)}. 
\end{array}$$
This shows that there is no deviating strategy that can result in the
player gaining
more than $\frac{1}{q(n)}$ in $G^n$ for $n>\max\{n_0,n_1\}$.  
\end{proof}

\subsection{Dealing with a variable number of players}\label{sec:graphical}
Up to now, we have assumed, just as in Borgs et al.~\citeyear{borgs2010myth},
that the number of players in
the game is a fixed constant ($\ge 3$).  

What happens if the
number of players in the game is part of the input?  
In general, describing the players' utilities in such a game takes   
space exponential in the number of players (since there are
exponentially many strategy profiles). 
Thus, to get interesting
computational results, we consider games that can be represented
succinctly.

Graphical games~\cite{KLS01} of degree $d$ are games that
can be represented by a graph in which each player is a node in the
graph, and the utility of a player is a function of only his action and
the actions of the players to which he is connected by an edge. The maximum
degree of a node is assumed to be at most $d$. This means a player's
punishment strategy depends only on the actions of at most $d$ players. 

\begin{definition}
Let $\G'_{a,b,d,n,m}$ be the set of all graphical games with degree at
most $d$, at most $m$ players and at most $n$ actions per player, 
integral payoffs,\footnote{Again, our result also holds for rational
payoffs, except then the size of the game needs to take into account the
bits needed to represent the payoffs.}
maximum payoff $a$, and minimum payoff $b$. 
\end{definition}

The following corollary then follows from our theorem, the fact that
a correlated equilibrium with polynomial sized-support can be computed
in polynomial time~\cite{jiang2011polynomial}, and the observation that
we can easily 
compute a correlated minimax strategy that depends only on the action
of at most $d$ players. 
\begin{cor}\label{VarPlayerEqExist}
For all $a$, $b$, $d$, and all polynomials $q$, there is a polynomial $f$
and a polynomial-time algorithm $F$
such that,
for all sequences $G_1, G_2, \ldots$ of games with $G^j\in G_{a,b,d,j,j}$
and for all inverse polynomials $\delta\leq1/f$,
the sequence of outputs of $F$ given the sequence $G_1, G_2,
\ldots$ of inputs is a $\frac{1}{q}$-equilibrium
for $G_1^{\infty}(\delta(1)), G_2^{\infty}(\delta(2)), \ldots$.
\end{cor}

\section{Computational subgame-perfect equilibrium}\label{sec:perfecteq}

\subsection{Motivation and Definition}
In this section we would like to define a notion similar to subgame-perfect equilibrium, where
for all histories $h$ in the game tree (even ones not on the equilibrium
path), playing $\vec{\sigma}$ restricted to the subtree starting at $h$
forms a NE.  
This means that a player does not have any incentive to deviate, no
matter where he finds himself in the game tree.

As we suggested in the introduction, there are a number of issues that
need to be addressed in
formalizing this intuition in our computational setting.  
First, since we consider stateful TMs, there is more to a descriptione
of a situation than just the history; we need to know the memory state
of the TM.  
That is, if we take a history to be just a sequence of actions, then the
analogue of history for us is really a pair $(h,\vec{m})$ consisting of
a sequence $h$ of actions, and a profile of memory states, one for each
player.  
Thus, to be a computational subgame-perfect equilibrium the strategies should be a NE at every history and no matter what the memory states are.



Another point of view is to say that
the players do not in fact have perfect information in our setting,
since we allow the TMs to have memory that is not observed by the other players, and thus the game should be understood as a game of imperfect information.
In a given history $h$ where $i$ moves, $i$'s
information set consists of all situations where the history is $h$ and
the states of memory of the other players are arbitrary.
While subgame-perfect equilibrium extends to imperfect information games 
it usually doesn't have much bite (see \cite{KW82} for a discussion on this point).
For the games that we consider, subgame-perfect equilibrium typically reduces to NE.
An arguably more natural generalization of subgame-perfect equilibrium in
imperfect-information games would require that if an
information set for 
player $i$ off the equilibrium path is reached, then player $i$'s
strategy is a best response to the other players' strategies \emph{no
matter how that information set is reached}. This is quite a strong
requirement. (see \cite{OR94}[pp.~219--221] for a discussion of this
issue); such equilibria do not in general exist in games of
imperfect information.%

Instead, in games of imperfect information, the solution concept most commonly
used is \emph{sequential equilibrium} \cite{KW82}.
A sequential equilibrium is a pair $(\vec{\sigma},\mu)$ consisting of a
strategy profile $\vec{\sigma}$ and a \emph{belief system} $\mu$, where
$\mu$ associates with each information set $I$ a probability $\mu(I)$ on
the nodes in $I$.  Intuitively, if $I$ is an information set for player $i$,
$\mu(I)$ describes $i$'s beliefs about the likelihood of being in each
of the nodes in $I$.
Then $(\vec{\sigma}, \mu)$ is a sequential equilibrium if, for each
player $i$ and each information set $I$ for player $i$, $\sigma_i$ is a
best response to $\vec{\sigma}_{-i}$ given $i$'s beliefs $\mu(I)$.
However, a common criticism of this solution concept is that it is unclear
what these beliefs should be and how players 	
create these beliefs. 
Instead, our notion of computational subgame-perfection can be viewed as a 
strong version of a sequential equilibrium, where, for each player $i$
and each information set $I$ for $i$, $\sigma_i$ is a best
response to $\vec{\sigma}_{-i}$ conditional on reaching $I$ (up to
$\epsilon$) no matter what $i$'s beliefs are at $I$.

As a deviating TM can change its memory state in arbitrary ways,
when we argue that a strategy profile is an
$\epsilon$-NE at a history, we must also consider all possible states
that the TM might start with at that history. Since there exists a
deviation that 
just rewrites the memory in the round just before the history we are
considering, any memory state (of polynomial length) is possible. 
Thus, in the computational setting, we require that the TM's
strategies are an 
$\epsilon$-NE at every history, no matter what the states of the TMs are
at that history.
This solution concept is in the spirit of subgame-perfect equilibrium,
as we require that the strategies are a NE after every possible
deviation, although the player might not have complete information as to
what the deviation is.

Intuitively, a profile $\vec{M}$ of TMs is a computational subgame-perfect
equilibrium if for all players $i$, all histories $h$ where $i$ moves, and
all memory profiles $\vec{m}$ of the players, there is no
polynomial-time TM $\bar{M}$ such that player $i$ can gain more
than $\epsilon$ by switching from $M_i$ to $\bar{M}$.  
To make it precise,
we must again consider an infinite sequence of games of increasing size
(just as we do for NE, although this definition is more complicated since we must consider memory
states). 

For a memory state $m$ and a TM $M$ let $M(m)$, stand for running $M$
with initial memory state $m$. 
We use $\vec{M}(\vec{m})$ to denote $(M_1(m_1),\ldots,M_c(m_c))$.
Let $p_i^{G,\delta}(\vec{M})$ denote player i's payoff in
$G^{\infty}(\delta)$ when $\vec{M}$ is played. 
\begin{definition} 
An infinite sequence of strategy profiles
$\vec{M}^1,\vec{M}^{2},\ldots$, where \mbox{$\vec{M}^k=(M^k_1,...,M^k_c)$}, is
a \emph{computational subgame-perfect $\epsilon$-equilibrium} of an
infinite sequence  
$G_1^{\infty},G_{2}^{\infty},\ldots$ of repeated games 
where the size of $G_k$ is $k$, if, for all players
$i\in [c]$, all sequences  \mbox{$h_1\in
H_{G_1^{\infty}},h_2\in H_{G_{2}^{\infty}},\ldots$} of histories,  
all sequences 
\mbox{$\vec{m}^1,\vec{m}^2,\ldots$} of polynomial-length memory-state
profiles, where 
\mbox{$\vec{m}^k=(m^k_1,\ldots,m^k_c)$}, 
and all non-uniform PPT
adversaries
$\bar{M}$ (polynomial in $k$ and 
$t$, as discussed above), 
there exists $k_0$ such that, for all $k\geq k_0$,
\begin{equation*}
p^{G_k^\infty(h_k),\delta}_i(\bar{M}(m^k_i),\vec{M}^k_{-i}(\vec{m}^k_{-i}))\leq
p^{G_k^\infty(h_k),\delta}_i(\vec{M}^k(\vec{m}^k))+\epsilon(k). 
\end{equation*}
\end{definition}


\subsection{Computing a subgame-perfect $\epsilon$-NE}

For a game $G$ such that $|G|=n$, and a polynomial $\ell$, consider the
following   
strategy $\sigma^{\NE,\ell}$,
and let $\vec{M}^{\sigma^{\NE,\ell}}$ be the TMs that implement this strategy.
This strategy is similar in spirit to that proposed
in Section~\ref{sec:computingNE}; indeed, the first two phases are identical. The key difference is that the punishment phase is played for only
$\ell(n)$ rounds.  After that, players return to phase 1.  As we show,
this limited punishment is effective since it is not played long enough
to make it an empty threat (if $\ell$ is chosen appropriately).  Phase
4 takes care of one minor issue: 
The fact that we can start in any memory state
means that a player might be called on to do something that, in fact, he
cannot do (because he doesn't have the information required to do it). 
For example, he might be called upon to play the correlated punishment
strategy in a state where he has forgotten the random seed, so he cannot
play it. In this case, a default action is played. Note that his was not an issue in the analysis of NE.

\begin{enumerate}
\item Play according to $\sq$ (with wraparound) as long as
all players played according to $\sq$ in the previous round.
\item After detecting a deviation by player $j \ne i$
in round $t_0$:%
\footnote{Again, if more than one player deviates while playing $\sq$, the
players punish the one with the smaller index. The punished player
plays his best response to what the other players are doing in this phase.} 
	\begin{enumerate}
\item Generate a pair $(pk_i,sk_i)$ using
$\Gen(1^n)$.
Store $sk_i$ in memory and 
use the next $l(n)$ rounds to broadcast $pk_i$,
as discussed above.
\item If $i = j+1$ (with wraparound), player $i$ does the following: 
\begin{itemize}
\item $i$ records $pk_{j'}$ for all players $j' \notin \{i,j\}$;
\item $i$ generates a random $n$-bit seed $\seed$;
\item for each player $j' \notin \{i,j\}$, $i$ computes
$m=Enc_{pk_{j'}}(\seed)$, 
and uses the  next \mbox{$(c-2)z(n)$} rounds to communicate these strings
to the players other than $i$ and $j$ (in some predefined order). 
\end{itemize}
\item If $i \ne j+1$, player $i$ does the following:
\begin{itemize}
\item $i$ records the actions played by $j+1$ at time slots designated
for $i$ to retrieve $Enc_{Pk_{i}}(\seed)$;
\item $i$ decrypts to obtain $\seed$, using $Dec$ and $sk_i$.
\end{itemize}
\end{enumerate}
\item  Phase 2 ends after $ \phi(n)+(c-2)z(n)$ rounds.  The players 
other than $j$ then
compute $\PS_{\seed}(t)$ and use it to determine which action profile to
play according to the distribution defined by   
a fixed (correlated) punishment strategy against $j$. 
Player $j$ plays his best response to the correlated
punishment strategy throughout this phase.
After $\ell(n)$ rounds, they return to phase 1, playing the sequence $\sq$
from the point at which the deviation occurred (which can easily be
inferred from the history).  
\item If at any point less than or equal to $ \phi(n)+(c-2)z(n)$ time steps
  from the last deviation from phase 1 the situation is incompatible
  with phase 2 as described above (perhaps because further deviations
  have occurred), or at any point between $\phi(n)+(c-2)z(n)$ and $\phi(n)+(c-2)z(n) +
  \ell(n)$ steps since the last deviation from phase 1 the situation
  is incompatible with phase 3 as 
  described above,  play a fixed action 
for the number of rounds left to complete phases 2 and 3 (i.e., up to
$\phi(n)+(c-2)z(n) + \ell(n)$ steps from the last deviation from phase 1).
Then return to phase 1.
\end{enumerate}
Note that with this strategy a deviation made during the punishment phase
is not punished. Phase $2$ and $3$ are always played to their full
length (which is fixed and predefined by $\ell$ and $z$).
We say that a history $h$ is a phase 1 history if it is a history where an
honest player should play according to $\sq$. 
History $h$ is a phase 2 history if it is a history where at most
$\phi(n) + (c-2)z(n)$ rounds have passed since the last deviation from phase 1; $h$
is a phase 3 history if more than $\phi(n) + (c-2)z(n)$ but at most
$\phi(n) + (c-2)z(n) + \ell(n)$ rounds have passed since the last
deviation from phase 1.   No matter what happens in phase $2$ and $3$,
a history in which exactly $\phi(n)+(c-2)z(n) + \ell(n)$ round have passed since the
last deviation from phase 1 is also a phase 1 history (even if the players deviate
from phase 2 and 3 in arbitrary ways).  Thus, no matter how many
deviations occur, we can uniquely identify the phase of each round. 

We next show that by selecting the right parameters, these strategies
are easy to compute and are a subgame-perfect $\epsilon$-equilibrium for
all inverse polynomials $\epsilon$. 

\begin{definition}
Let $\G_{a,b,c,n}$ be the set of all games with $c$ players, at most $n$
actions per player, 
integral payoffs,
maximum payoff $a$, and minimum payoff $b$.
\end{definition}

We first show that for any strategy that deviates while phase
1 is played, there is a strategy whose payoff is at least as good and
either does not deviate in the first polynomially many rounds, or after its
first deviation, deviates every time phase 1 is played. (Recall that
after every deviation in phase 1, the other players play the
punishment phase for $\ell(n)$ rounds and then play phase 1 again.)

We do
this by showing that if player $i$ has a profitable deviation at some round
$t$ of phase 1, then it must be the case that every time this round of
phase 1 is played, $i$ has a profitable deviation there.  
(That is, the strategy of deviating every time this round of phase 1 is
played is at least as good as a strategy where player $i$ correlates his
plays in different instantiations of phase 1.)  
While this is trivial in traditional game-theoretic analyses,
naively applying it in the computational setting does not necessarily
work. It requires us to formally show how we 
reduce a polynomial time TM $M$ to a different TM $M$'
of the desired form without blowing up the running time and size of the TM.

For a game $G$, let $H_{G^{\infty}}^{1,n,f}$ be the set of histories 
$h$ of $G^{\infty}$ of length at most $nf(n)$ such that
at (the last node of) $h$, $\sigma^{\NE,\ell}$ is in phase 1.
Let $R(M)$ be the polynomial that bounds the running time of TM $M$.

\begin{definition}
Given a game $G$, a deterministic TM $M$ is said to be
$(G,f,n)$-\emph{well-behaved} 
if, when $(M,\sigma^{\NE,\ell}_{-i})$ is played, then either $M$ does
not deviate for the first $nf(n)$ rounds or, after
$M$ first deviates,
$M$ continues to deviate from $\sq$ every time phase 1 is
played in the next $nf(n)$ rounds. 
\end{definition}

\begin{lemma}\label{alwaysDeviate}
For all $a$, $b$, $c$, and all polynomials $f$, there exists a polynomial
$g$ such that for all $n$, all games $G\in 
\G_{a,b,c,n}$, all $h\in H_{G^{\infty}}^{1,n,f}$, all
players $i$, 
and all TMs $M$, there exists a
$(G(h),f,n)$-well-behaved TM M' such that 
$p_i^{G^h,1/f(n)}(M',\vec{M}^{\sigma^{\NE,\ell}}_{-i})\geq p_i^{G^h,1/f(n)}(M,\vec{M}^{\sigma^{\NE,\ell}}_{-i})$,
and $R(M'),|M'|\leq g(R(M))$.
\end{lemma}

\begin{proof}
Suppose that we are given $G \in \G_{a,b,c,n}$, $h\in
H_{G^{\infty}}^{1,n,f}$, and a TM $M$.
We can assume without loss of generality that $M$ is
deterministic
(we can always just use the best random tape).
If $M$ does not deviate in the first $nf(n)$ rounds of $G(h)^{\infty}$
then $M'$ is 
just $M$, and we are done.  
Otherwise, we construct a sequence of TMs starting with $M$ that are, in a
precise sense, more and more well behaved, until eventually we get the
desired TM $M'$.

For $t_1 < t_2$, say that $M$ is
\emph{$(t_1,t_2)$-$(G,f,n)$-well-behaved} if $M$ does not deviate from $\sq$
until round $t_1$, and then deviates from $\sq$ every time phase 1 is
played up to (but not including) round $t_2$
(by which we mean there exists some history in which $M$ does not
deviate at round $t_2$ and this is the shortest such history
over all possible random tapes of $\vec{M}^{\sigma^{\NE,\ell}}_{-i}$).
We construct a sequence $M_1, M_2, \ldots$ of TMs such that (a) $M_1 = M$,
(b) $M_i$ is $(t_1^i,t_2^i)$-$(G,f,n)$-well-behaved, 
(c) either $t_1^{i+1} >t_i$ or $t_1^{i+1} = t_1^i$ and $t_2^{i+1} > t_2^{i}$, 
and
(d) $p_i^{G^h,1/f(n)}(M_{i+1},\vec{M}^{\sigma^{\NE,\ell}}_{-i}) \ge
p_i^{G^h,1/f(n)}(M_i,\vec{M}^{\sigma^{\NE,\ell}}_{-i})$.
Note that if
$t_1 \ge nf(n)$ or $t_2 \ge t_1 + nf(n)$, then a
$(t_1,t_2)$-$(G,f,n)$-well-behaved TM is $(G,f,n)$-well-behaved.

Let $t<nf(n)$ be the first round at which $M$ deviates.
(This is well defined since the play up to $t$ is deterministic.) 
Let the history up to time $t$ be $h^t$.
If $M$ deviates every time that phase 1 is played for the $nf(n)$ rounds
after round $t$, then again we can take $M' = M$, and we are done.  
If not, let $t'$ be the first round after $t$ at which phase 1 is played and 
there exists some history of length $t'$ at which
$M$ does not deviate.  
By definition, $M$ is $(t,t')$-$(G,f,n)$-well
behaved. We take $M_1 = M$ and $(t_1^1,t_2^1) = (t,t')$. 
(Note that since $\vec{M}^{\sigma^{\NE,\ell}}_{-i}$ are randomized during phase 2, the first
time after $t$ at 
which $M$ returns to playing phase 1 and does not deviate may depend
on the results of their coin tosses.  We take $t'$ to be the first time this
happens with positive probability.)

Let $s^{h^*}$ be $M$'s memory state at a history $h^*$.
We assume for ease of exposition that $M$ encodes the history in its
memory state.  (This can be done, since the memory state at time $t$ 
is of size polynomial in $t$.) 
Consider the TM $M''$ that acts like $M$ up to round $t$, and copies $M$'s
memory state at that round (i.e., $s^{h^t}$). 
$M''$ continues to plays like $M$ up to the first round $t'$ with $t <
t' < t
+nf(n)$ at which
$\sigma^{NE,\ell}$ would be about to return to phase 1 
and $M$ does not deviate 
(which means
that $M$ plays an action in the sequence $\sq$ at round $t'$).  
At round $t'$, $M''$ sets its state to
$s^{h^t}$ and simulates $M$ from history $h^t$ with states $s^{h(t)}$;
so, in particular, $M''$ does deviate at time $t'$.  (Again, the time
$t'$ may depend on random choices made by $\vec{M}^{\sigma^{\NE,\ell}}_{-i}$.  We assume that $M''$
deviates the first time $M$ is about to play phase 1 after round $t$
and does not deviate, no matter what the outcome of the coin tosses.)
This means, in particular, that $M''$ deviates at any such $t'$.
We call $M''$ a \emph{type 1 deviation from $M$}.

If $p_i^{G^{h^t},1/f(n)}(M'',\vec{M}^{\sigma^{\NE,\ell}}_{-i}) >
p_i^{G^{h^t},1/f(n)}(M,\vec{M}^{\sigma^{\NE,\ell}}_{-i})$, then we
take $M_2 = M''$.  Note that $t^2_1 =t^1_1 = t$, while $t^2_2 > t^1_2 =
t'$, since $M''$ deviates at $t'$.
If $p_i^{G^{h^t},1/f(n)}(M'',\vec{M}^{\sigma^{\NE,\ell}}_{-i}) <
p_i^{G^{h^t},1/f(n)}(M,\vec{M}^{\sigma^{\NE,\ell}}_{-i})$,  
then there exists 
some history $h^*$ of both $M$ and $M''$ such that $t < |h^*|<t+nf(n)$,
$M''$ deviates at $h^*$, 
$M$ does not,
and $M$ has a better expected payoff than $M''$ at $h^*$.
(This is a history where the type 1 deviation failed to improve the payoff.)
Take $M_2$ to be the TM that plays like $\vec{M}^{\sigma^{\NE,\ell}}_i$ up to time
$t$, then 
sets its state to $s^{h^{*}}$, and then plays like $M$ with state
$s^{h^{*}}$ in history $h^{*}$.  
We call $M_2$  a \emph{type 2 deviation from $M$}.
Note that $M_2$ does
not deviate at $h^t$ (since $M$ did not deviate at history $h^{*}$).
Let $\delta'=(1-\delta)^{|h^*|-|h^t|}$.
Clearly $\delta' p_i^{G^{h^t},1/f(n)}(M_2,\vec{M}^{\sigma^{\NE,\ell}}_{-i}) = 
p_i^{G^{h^{*}},1/f(n)}(M,\vec{M}^{\sigma^{\NE,\ell}}_{-i})$,
since $\vec{M}^{\sigma^{\NE,\ell}}_{-i}$ acts the same in $G^{h^t}$ and
$G^{h^{*}}$. 
Since $M''$ plays like $M(s^{h_t})$ at $h^*$,  
$p_i^{G^{h^{*}},1/f(n)}(M'',\vec{M}^{\sigma^{\NE,\ell}}_{-i})=\delta'
p_i^{G^{h^t},1/f(n)}(M,\vec{M}^{\sigma^{\NE,\ell}}_{-i})$. Combining
this with the previous observations, we get that
$p_i^{G^{h^t},1/f(n)}(M_2,\vec{M}^{\sigma^{\NE,\ell}}_{-i}) \geq  
p_i^{G^{h^{t}},1/f(n)}(M,\vec{M}^{\sigma^{\NE,\ell}}_{-i})$.
Also note that $t_1^2>t_1^1$. 
This completes the construction of $M_2$.  
We inductively construct $M_{i+1}$, $i = 2, 3, \ldots$, just as we did $M_2$,
letting $M_i$ play the role of $M$.

Next observe that, without loss of generality, we can assume that
this sequence arises from a sequnce of type $2$ 
deviations, followed by a sequence of type $1$ deviations:
For let $j_1$ be the first point in the sequence at which a type $1$
deviation is made.  We claim that we can assume without loss of
generality that all further deviations are type 1 deviations.  
By assumption, since $M_{j_1}$ gives $i$ higher utility than $M_{j_1 -  1}$, 
it is better to deviate the first time $M_{j_1-1}$ wants to
play phase 1 again after an initial deviation. This means that when
$M_{j_1}$ wants to play phase 1 again after an initial deviation it
must be better to deviate again, since the future play of the
$\vec{M}^{\sigma^{\NE,\ell}}_{-i}$ is the same in both of these
situations. 
This means that once a type $1$ deviation occurs, we can assume that
all further deviations are type 1 deviations.

Let $M_j$ be the first TM in the sequence that is well behaved.  (As
we observed earlier, there must be such a TM.)  
Using the fact that the sequence consists of a sequence of type 2
deviations followed by a sequence of type 1 deviations, 
it is not hard to show that $M_j$ can be implemented efficiently. First
notice that $M_{j_1}$ is a TM that plays like
$\vec{M}^{\sigma^{\NE,\ell}}_i$ until some round, and then plays $M$
starting with its state at a history which is at most $(nf(n))^2$
longer than the real history at this point. This is because its
initial history becomes longer by at most $nf(n)$ at each round and we
iterate this construction at most $nf(n)$ times. This means that its
running time is obviously polynomially related to the running time of
the original $M$. 
The same is true of the size of $M_{j_1}$, since we need to 
encode only the state at this initial history and the history at which we
switch, which is polynomially related to $R(M)(n)$. 

To construct $M_j$, we need to modify $M_{j_1}$ only slightly, since
only type $1$ deviations occur. Specifically, we need to
know only $t_{j_1}^1$ and to encode its state at this round. At every history
after that, we run $M_{J_1}$ (which is essentially running $M$ on a
longer history) on a fixed history, with a potential additional step of
copying the state. It is easy to see that the resulting TM has running time and
size at most $O(R(M))$.  
\end{proof}

We now state and prove our theorem, which shows that there exists a
polynomial-time algorithm for computing a subgame-perfect
$\epsilon$-equilibrium by showing that, for all inverse 
polynomials $\epsilon$, 
there exists a polynomial function $\ell$ of $\epsilon$ such that 
$\sigma^{\NE^*,\ell}$ is a 
subgame-perfect $\epsilon$-equilibrium of the game.
The main idea of the proof is to show that the players can't gain much
from deviating while the sequence is being played, and also that, since
the punishment is relatively short, deviating while a player is being
punished is also not very profitable.  

\begin{theorem}\label{perfectEqExist}
For all $a$, $b$, $c$, and all polynomials $q$, there is a polynomial $f$
and a polynomial-time algorithm 
$F$ such that, 
for all sequences $G_1, G_2, \ldots$ of games with $G^j\in G_{a,b,c,j}$
and for all inverse polynomials $\delta\leq1/f$,
the sequence of outputs of $F$ given the sequence $G_1, G_2,
\ldots$ of inputs is a subgame-perfect $\frac{1}{q}$-equilibrium
for $G_1^{\infty}(\delta(1)), G_2^{\infty}(\delta(2)), \ldots$.
\end{theorem}
\newenvironment{proofsketch}{\trivlist\item[]\emph{Proof Sketch}:}%
{\unskip\nobreak\hskip 1em plus 1fil\nobreak$\Box$
\parfillskip=0pt%
\endtrivlist}
\begin{proof}
Given a game $G^n \in \G(a,b,c,n)$,
the algorithm finds a correlated equilibrium
$\sigma$ of $G^n$, which can be done in polynomial time using linear
programming. Each player's expected payoff is at least
$0$ when playing $\sigma$, since we assumed that the minimax value of the game
is $0$. 
Let $r=a-b$.  By Lemma~\ref{sequence} and Lemma~\ref{derandomize}, we can construct a sequence $\sq$ of
length $w(n)=4(r n q(n)+1)n^c$ and set $f'(n)=4r w(n)q(n)$, so that if
the players play $\sq$ infinitely often and $\delta<1/f'(n)$, then all the
players get at least $-1/2q(n)$.  
The correlated punishment strategy against each player can also be found in polynomial time
using linear programming. 

Let $m(n)$ be the length of phase 2, including the round where the
deviation occurred.    
(Note that $m(n)$ is a polynomial that depends
only on the choice of encryption scheme---that is, it depends on $\phi$,
where a $\phi$-bit public-key encryption scheme is used, and on $z$, where
$z(k)$ is the length of encrypted messages.) 
Let $\ell(n)=nq(n)(m(n)a+1)$, let $\sigma^{*}_n$ be the strategy
$\vec{M}^{\sigma^{\NE,\ell}}$ described 
above, and let \mbox{$f(n)=\max(3rq(n)(\ell(n)+m(n)),f'(n))$}.

We now show that $\sigma^{*}_{1},\sigma^{*}_{2},\ldots$ is a subgame-perfect
$(1/q)$-equilibrium 
for every inverse polynomial discount factor $\delta\leq 1/f$. 
We focus on deviations at histories of length $<\frac{n}{\delta(n)}$, since, by
Lemma \ref{negligible}, the sum of payoffs received after that is
negligible.  Thus, there exists some $n_0$ such that, for all $n>n_0$,
the payoff achieved 
after that history is less than $1/q(n)$, which does not justify deviating.

We first show that no player has an incentive to deviate in subgames
starting from phase 1 histories.
By Lemma \ref{alwaysDeviate}, it suffices to consider
only a deviating strategy 
that after its first deviation deviates every time phase 1 is played;
for every deviating strategy, either not deviating does at least as well
or there is a deviating strategy of this form that does at least as well.
Let $h_1$ be the history in which the deviation occurs and let $M$ be the deviating strategy. Notice that
$\vec{M}^{\sigma^{\NE,\ell}}$ can always act as intended at such
histories; it can detect it is in such a history and can use the
history to compute the next move (i.e., it does not need to maintain
memory to figure out what to do next).

The player's payoff from $(M,\vec{M}^{\sigma^{\NE,\ell}}_{-i})$ during
one cycle of deviation and punishment can be at 
most $a$ at each round of phase 2 and, by Lemma~\ref{psfMinmax}, is
negligible throughout phase 3.
(We use $\epsilon_{neg}$ to denote the negligible payoff to a deviator
in phase 3.)
Thus, the payoff of the deviating player from
$(M,\vec{M}^{\sigma^{\NE,\ell}}_{-i})$ from the point of deviation
onwards is at most 
\begin{align*}
((1-\delta(n)^{|h_1|})\big(\delta(n)(m(n)a+\epsilon_{neg})\sum_{t=0}^{\lceil\frac{nf(n)-|h_1|}{m(n)+\ell(n)}\rceil}(1-\delta(n))^{(m(n)+\ell(n))t}+\epsilon'_{neg}\big)\\
\leq ((1-\delta(n)^{|h_1|})\big(\delta(n)(m(n)a+\epsilon_{neg})\sum_{t=0}^{\infty}(1-\delta(n))^{(m(n)+\ell(n))t}+\epsilon'_{neg}\big),
\end{align*}
where $\epsilon'_{neg}$ is the expected payoff after round $nf(n)$.
By Lemma~\ref{sequence}, no matter where in the sequence the
players are, the average discounted payoff at that point from playing honestly is at least
$-1/2q(n)$. Thus, the payoff from playing
$(\vec{M}^{\sigma^{\NE,\ell}})$ from this point onwards is at least 
$-(1-\delta(n))^{|h_1|})1/2q(n).$
We can ignore any payoff before the deviation since it is the same
whether or not the player deviates.
and also divide both sides by $(1-\delta(n))^{|h_1|})$; thus,
it suffices to prove that
$$\delta(n)(m(n)a+\epsilon_{neg})\sum_{t=0}^{\infty}(1-\delta(n))^{(m(n)+\ell(n))t}+\epsilon'_{neg}\leq 
\frac{1}{2q(n)}.$$ 
The term on the left side is bounded by
$O\big(\frac{m(n)a+\epsilon_{neg}}{nq(n)(m(n)a+1)}\big)$, and thus there
exists $n_1$ such that, for all $n>n_1$, the term on the left side is
smaller than 
$\frac{1}{2q(n)}$
(In fact, for all constants $c$, there exists $n_c$ such that 
the left-hand side is at most $\frac{1}{cq(n)}$ for any $n > n_c$.)

We next show that no player wants to deviate in phase 2 or 3 histories.
Notice that since these phases are carried out to completion even if
the players deviate while in these phases (we do not punish them for
that), and the honest strategy can easily detect whether it is in such a
phase by looking at when the last deviation from phase 1
occurred. 
First consider a punishing player. By not following the strategy, he can
gain at most $r$ for at most $\ell(n)+m(n)$ rounds over the payoff he
gets with
the original strategy 
(this is true even if his memory state is such that he just plays a
fixed action, or even if another player deviates while the phase is played).  
Once the players start playing phase 1 again,
our previous claim shows that no matter what the actual history is at
that point, a strategy that does not follow the sequence does not gain
much. 
It is easy to verify that, given the discount factor, a
deviation can increase his discounted payoff by at most
$\frac{1}{q(n)}$ in this case. 
(Notice that the previous claim works for any constant fraction of
$1/q(n)$, which is what we are using here since the deviation in the
punishment phase gains $1/cq(n)$ for some $c$.)

The punished player can deviate to a TM that correctly guessed
the keys chosen 
(or the current TM's memory state might contain the actual keys and he
defects to a
TM that uses these keys) , 
in which case he would know exactly what the players are going to
do while they are punishing him. Such a deviation exists once the keys
have been played and are part of the history.
Another deviation might be a result of the other TMs being in an
inconsistent memory state, so that they play a fixed action, one which
the deviating player might be able to take advantage of. 
However, these deviations work (or any other possible deviation) only for the current punishment phase.
Once the players go back to playing phase 1, 
this player can not gain much by deviating from the sequence
again. 
For if he deviates again, the other players will choose new random
keys and a new random seed
(and will have a consistent memory state); 
from our previous claims, this means that no strategy can gain
much over a strategy that follows the sequence.  
Moreover, he can also 
gain at most $r$ for at most $\ell(n)+m(n)$ rounds which, as claimed
before, means that his discounted payoff difference is less than
$\frac{1}{q(n)}$ in this case.  

This shows that, for $n$ sufficiently large, no player can gain more than
$1/q(n)$ from deviating at any history.  Thus, this strategy is 
a subgame-perfect $1/q$-equilibrium.
\end{proof}

Using the same arguments as in Section~\ref{sec:graphical}, we can
also apply these ideas to efficiently find a computational
subgame-perfect $\epsilon$-equilibrium in constant-degree graphical games.

\begin{cor}\label{VarPlayerPerEqExist}
For all $a$, $b$, $d$, and all polynomials $q$, there is a polynomial $f$
and a polynomial-time algorithm $F$
such that,
for all sequences $G_1, G_2, \ldots$ of games with $G^j\in G_{a,b,d,j,j}$
and for all inverse polynomials $\delta\leq1/f$,
the sequence of outputs of $F$ given the sequence $G_1, G_2,
\ldots$ of inputs is a subgame-perfect $\frac{1}{q}$-equilibrium
for $G_1^{\infty}(\delta(1)), G_2^{\infty}(\delta(2)), \ldots$.
\end{cor}



%

\section{Acknowledgments}
Joseph Halpern and Lior Seeman are supported in part by NSF grants
IIS-0911036 and CCF-1214844, by	AFOSR	grant	FA9550-08-1-0266,	
by	ARO grant W911NF-14-1-0017, and by the Multidisciplinary
University Research Initiative (MURI) program administered by the
AFOSR under grant FA9550-12-1-0040.
Lior Seeman is partially supported by a grant from the Simons Foundation \#315783.  
Rafael Pass is supported in part by an Alfred P. Sloan Fellowship, a
Microsoft Research Faculty Fellowship, 
NSF Awards CNS-1217821 and CCF-1214844, NSF CAREER Award CCF-0746990,
AFOSR YIP Award FA9550-10-1-0093, and DARPA and AFRL under contract
FA8750-11-2-0211. 
The views and conclusions contained in this document
are those of the authors 
and should not be interpreted as representing the official policies,
either expressed or implied, of the Defense Advanced Research Projects
Agency or the US Government.

\bibliographystyle{chicagor}
\bibliography{TruthBehindTheMyth,joe,z}

\begin{thebibliography}{}

\bibitem[\protect\citeauthoryear{Andersen and Conitzer}{Andersen and
  Conitzer}{2013}]{fast2013Andersen}
Andersen, G. and V.~Conitzer (2013).
\newblock Fast equilibrium computation for infinitely repeated games.
\newblock In {\em Twenty-Seventh AAAI Conference on Artificial Intelligence}.

\bibitem[\protect\citeauthoryear{Aumann and Shapley}{Aumann and
  Shapley}{1994}]{AS94}
Aumann, R. and L.~Shapley (1994).
\newblock Long-term competition—a game-theoretic analysis.
\newblock In N.~Megiddo (Ed.), {\em Essays in Game Theory}, pp.\  1--15.
  Springer New York.

\bibitem[\protect\citeauthoryear{Borgs, Chayes, Immorlica, Kalai, Mirrokni, and
  Papadimitriou}{Borgs et~al.}{2010}]{borgs2010myth}
Borgs, C., J.~Chayes, N.~Immorlica, A.~Kalai, V.~Mirrokni, and C.~Papadimitriou
  (2010).
\newblock The myth of the folk theorem.
\newblock {\em Games and Economic Behavior\/}~{\em 70\/}(1), 34--43.

\bibitem[\protect\citeauthoryear{Chen and Deng}{Chen and
  Deng}{2006}]{chen2006settling}
Chen, X. and X.~Deng (2006).
\newblock Settling the complexity of two-player nash equilibrium.
\newblock In {\em Proc.~47th IEEE Symposium on Foundations of Computer
  Science}, pp.\  261--272.

\bibitem[\protect\citeauthoryear{Chen, Deng, and Teng}{Chen
  et~al.}{2006}]{chen2006computing}
Chen, X., X.~Deng, and S.~Teng (2006).
\newblock Computing {N}ash equilibria: Approximation and smoothed complexity.
\newblock In {\em Proc.~47th IEEE Symposium on Foundations of Computer
  Science}, pp.\  603--612.

\bibitem[\protect\citeauthoryear{Daskalakis, Goldberg, and
  Papadimitriou}{Daskalakis et~al.}{2006}]{daskalakis2006complexity}
Daskalakis, C., P.~Goldberg, and C.~Papadimitriou (2006).
\newblock The complexity of computing a nash equilibrium.
\newblock In {\em Proc.~38th ACM Symposium on Theory of Computing}, pp.\
  71--78.

\bibitem[\protect\citeauthoryear{Diffie and Hellman}{Diffie and
  Hellman}{1976}]{diffie1976new}
Diffie, W. and M.~Hellman (1976).
\newblock New directions in cryptography.
\newblock {\em IEEE Transactions on Information Theory\/}~{\em 22\/}(6),
  644--654.

\bibitem[\protect\citeauthoryear{Dodis, Halevi, and Rabin}{Dodis
  et~al.}{2000}]{dodis2000cryptographic}
Dodis, Y., S.~Halevi, and T.~Rabin (2000).
\newblock A cryptographic solution to a game theoretic problem.
\newblock In {\em CRYPTO 2000: 20th International Cryptology Conference}, pp.\
  112--130.

\bibitem[\protect\citeauthoryear{Fudenberg and Maskin}{Fudenberg and
  Maskin}{1986}]{fudenberg1986folk}
Fudenberg, D. and E.~Maskin (1986).
\newblock The folk theorem in repeated games with discounting or with
  incomplete information.
\newblock {\em Econometrica\/}~{\em 54\/}(3), 533--554.

\bibitem[\protect\citeauthoryear{Goldreich}{Goldreich}{2001}]{goldreichFound}
Goldreich, O. (2001).
\newblock {\em Foundation of Cryptography, Volume I Basic Tools}.

\bibitem[\protect\citeauthoryear{Goldreich, Goldwasser, and Micali}{Goldreich
  et~al.}{1986}]{goldreich1986construct}
Goldreich, O., S.~Goldwasser, and S.~Micali (1986).
\newblock How to construct random functions.
\newblock {\em Journal of the ACM\/}~{\em 33\/}(4), 792--807.

\bibitem[\protect\citeauthoryear{Goldwasser and Micali}{Goldwasser and
  Micali}{1984}]{goldwasser1984probabilistic}
Goldwasser, S. and S.~Micali (1984).
\newblock Probabilistic encryption.
\newblock {\em Journal of Computer and System Sciences\/}~{\em 28\/}(2),
  270--299.

\bibitem[\protect\citeauthoryear{Gossner}{Gossner}{1998}]{gossner1998repeated}
Gossner, O. (1998).
\newblock {\em Repeated games played by cryptographically sophisticated
  players}.
\newblock Center for Operations Research \& Econometrics. Universit{\'e}
  catholique de Louvain.

\bibitem[\protect\citeauthoryear{Gossner}{Gossner}{2000}]{gossner2000sharing}
Gossner, O. (2000).
\newblock Sharing a long secret in a few public words.
\newblock Technical report, THEMA (TH{\'e}orie Economique, Mod{\'e}lisation et
  Applications), Universit{\'e} de Cergy-Pontoise.

\bibitem[\protect\citeauthoryear{Gradwohl, Livne, and Rosen}{Gradwohl
  et~al.}{2013}]{GLR13}
Gradwohl, R., N.~Livne, and A.~Rosen (2013, January).
\newblock Sequential rationality in cryptographic protocols.
\newblock {\em ACM Trans. Econ. Comput.\/}~{\em 1\/}(1), 2:1--2:38.

\bibitem[\protect\citeauthoryear{Halpern and Pass}{Halpern and
  Pass}{2013}]{HP13}
Halpern, J.~Y. and R.~Pass (2013).
\newblock Sequential equilibrium in computational games.
\newblock In {\em Proc.~ 23rd International Joint Conference on Artificial
  Intelligence (IJCAI '13)}, pp.\  171--176.

\bibitem[\protect\citeauthoryear{H{\aa}stad, Impagliazzo, Levin, and
  Luby}{H{\aa}stad et~al.}{1999}]{haastad1999pseudorandom}
H{\aa}stad, J., R.~Impagliazzo, L.~A. Levin, and M.~Luby (1999).
\newblock A pseudorandom generator from any one-way function.
\newblock {\em SIAM Journal on Computing\/}~{\em 28\/}(4), 1364--1396.

\bibitem[\protect\citeauthoryear{Hoeffding}{Hoeffding}{1963}]{hoeffding1963probability}
Hoeffding, W. (1963).
\newblock Probability inequalities for sums of bounded random variables.
\newblock {\em Journal of the American Statistical Association\/}~{\em
  58\/}(301), 13--30.

\bibitem[\protect\citeauthoryear{Jiang and Leyton-Brown}{Jiang and
  Leyton-Brown}{2011}]{jiang2011polynomial}
Jiang, A.~X. and K.~Leyton-Brown (2011).
\newblock Polynomial-time computation of exact correlated equilibrium in
  compact games.
\newblock In {\em Proceedings of the 12th ACM conference on Electronic
  commerce}, pp.\  119--126. ACM.

\bibitem[\protect\citeauthoryear{Kearns, Littman, and Singh}{Kearns
  et~al.}{2001}]{KLS01}
Kearns, M., M.~L. Littman, and S.~P. Singh (2001).
\newblock Graphical models for game theory.
\newblock In {\em Proc.~Seventeenth Conference on Uncertainty in Artificial
  Intelligence (UAI 2001)}, pp.\  253--260.

\bibitem[\protect\citeauthoryear{Kol and Naor}{Kol and Naor}{2008}]{KN08a}
Kol, G. and M.~Naor (2008).
\newblock Games for exchanging information.
\newblock In {\em Proc.~40th Annual ACM Symposium on Theory of Computing (STOC
  '08)}, pp.\  423--432.

\bibitem[\protect\citeauthoryear{Kreps and Wilson}{Kreps and
  Wilson}{1982}]{KW82}
Kreps, D.~M. and R.~B. Wilson (1982).
\newblock Sequential equilibria.
\newblock {\em Econometrica\/}~{\em 50}, 863--894.

\bibitem[\protect\citeauthoryear{Lehrer}{Lehrer}{1991}]{lehrer1991internal}
Lehrer, E. (1991).
\newblock Internal correlation in repeated games.
\newblock {\em International Journal of Game Theory\/}~{\em 19\/}(4), 431--456.

\bibitem[\protect\citeauthoryear{Littman and Stone}{Littman and
  Stone}{2005}]{littman2005polynomial}
Littman, M.~L. and P.~Stone (2005).
\newblock A polynomial-time nash equilibrium algorithm for repeated games.
\newblock {\em Decision Support Systems\/}~{\em 39\/}(1), 55--66.

\bibitem[\protect\citeauthoryear{Neyman}{Neyman}{1985}]{neyman1985}
Neyman, A. (1985).
\newblock Bounded complexity justifies cooperation in the finitely repeated
  prisoners' dilemma.
\newblock {\em Economics Letters\/}~{\em 19\/}(3), 227--229.

\bibitem[\protect\citeauthoryear{Osborne and Rubinstein}{Osborne and
  Rubinstein}{1994}]{OR94}
Osborne, M.~J. and A.~Rubinstein (1994).
\newblock {\em A Course in Game Theory}.
\newblock \chicagoraddresspub{Cambridge, Mass.: }MIT Press.

\bibitem[\protect\citeauthoryear{Papadimitriou and Yannakakis}{Papadimitriou
  and Yannakakis}{1994}]{PY94}
Papadimitriou, C.~H. and M.~Yannakakis (1994).
\newblock On complexity as bounded rationality.
\newblock In {\em Proc.~26th ACM Symposium on Theory of Computing}, pp.\
  726--733.

\bibitem[\protect\citeauthoryear{Rivest, Shamir, and Adleman}{Rivest
  et~al.}{1978}]{rivest1978method}
Rivest, R.~L., A.~Shamir, and L.~Adleman (1978).
\newblock A method for obtaining digital signatures and public-key
  cryptosystems.
\newblock {\em Communications of the ACM\/}~{\em 21\/}(2), 120--126.

\bibitem[\protect\citeauthoryear{Rubinstein}{Rubinstein}{1979}]{Ru79}
Rubinstein, A. (1979).
\newblock Equilibrium in supergames with the overtaking criterion.
\newblock {\em Journal of Economic Theory\/}~{\em 21\/}(1), 1--9.

\bibitem[\protect\citeauthoryear{Rubinstein}{Rubinstein}{1986}]{rubinstein1986finite}
Rubinstein, A. (1986).
\newblock Finite automata play the repeated prisoner's dilemma.
\newblock {\em Journal of Economic Theory\/}~{\em 39\/}(1), 83--96.

\bibitem[\protect\citeauthoryear{Selten}{Selten}{1965}]{Selten65}
Selten, R. (1965).
\newblock Spieltheoretische behandlung eines oligopolmodells mit
  nachfragetr\"{a}gheit.
\newblock {\em Zeitschrift f\"{u}r Gesamte Staatswissenschaft\/}~{\em 121},
  301--324 and 667--689.

\bibitem[\protect\citeauthoryear{Urbano and Vila}{Urbano and
  Vila}{2004}]{urbano2004unmediated}
Urbano, A. and J.~Vila (2004).
\newblock Unmediated communication in repeated games with imperfect monitoring.
\newblock {\em Games and Economic Behavior\/}~{\em 46\/}(1), 143--173.

\bibitem[\protect\citeauthoryear{Urbano and Vila}{Urbano and
  Vila}{2002}]{urbanoVilla2002}
Urbano, A. and J.~E. Vila (2002).
\newblock Computational complexity and communication: Coordination in
  two-player games.
\newblock {\em Econometrica\/}~{\em 70\/}(5), 1893--1927.

\end{thebibliography}

\appendix
\section*{APPENDIX}
\section{Multi-Instance PRFs}
\label{section:appA}
In this section, we show that for any family of PRF, even polynomially many
random members of it are indistinguishable from polynomially many truly
random functions. 
\begin{lemma}
For all polynomials $q$, if $\{f_s: \{0,1\}^{|s|}\to
\{0,1\}^{|s|}\}_{s\in\{0,1\}^*}$ is a pseudorandom function ensemble,
then the  ensemble $F^q=\{F_n^1,\ldots, F_n^{q(n)}\}_{n\in\mathbb{N}}$
where, for all $i$, $F_n^i$ is uniformly distributed over the multiset
$\{f_s\}_{s\in\{0,1\}^n}$, is computationally indistinguishable from
$H^q=\{H_n^1,\ldots,H_n^{q(n)}\}_{n\in\mathbb{N}}$. 
\end{lemma}
\begin{proof}
Assume for contradiction that the ensembles are distinguishable. This means there exist a polynomial $q$, a PPT $D$, and a polynomial $p$ such that for infinitely many $n$'s 
$$|Pr[D(1^n,(H_n^1,\ldots,H_n^{q(n)}))=1]-Pr[D(1^n,(F_n^1,\ldots,F_n^{q(n)}))=1]|>\frac{1}{p(n)}.$$

For each $n$, let
$T_n^i=(1^n,(H_n^1,\ldots,H_n^{i-1},F_n^{i},\ldots,F_n^{q(n)}))$.  
We can now describe a PPT $D'$ that distinguishes
$\{F_n\}_{n\in\mathbb{N}}$ and $\{H_n\}_{n\in\mathbb{N}}$ for infinitely
many $n$'s. First notice that a PPT can easily simulate polynomially many
oracle queries to both a truly random function and to a member of
$F_n$. So $D'$ on input $(1^n,X)$ randomly chooses
$j\in\{1,\ldots,q(n)\}$ and calls $D$ with input
$(1^n,(I^1,\ldots,I^{j-1},X,J^{j+1},\ldots,J^{q(n)}))$, where it
simulates a query to $I_k$ as a query to a random member of $H_n$, and a
query to $J_k$ as a query to a random member of $F_n$.  (Notice that since
$D$ is a PPT, it can make only polynomially many oracle queries to any of
the functions, which can be easily simulated). Whenever $D$
makes an oracle query to $X$, $D'$ makes an oracle query to $X$, and uses
its answer as the answer to $D$. 
When $D$ terminates, $D'$ outputs the same value as $D$. 

Now notice that if $X$ is $H_n$, then the input to $D$ is $T_n^j$, while
if $X$ is $F_n$, then the input to $D$ is $T_n^{j+1}$.  Thus, 
$Pr[D'(1^n,H_n)=1]=\frac{1}{q(n)}\sum_{i=1}^{q(n)}Pr[D(T_n^{i+1})=1]$,
and \\
\mbox{$Pr[D'(1^n,F_n)=1]=\frac{1}{q(n)}\sum_{i=1}^{q(n)}Pr[D(T_n^{i})=1]$}.
It follows that
\begin{align*}
|Pr[D'(1^n,H_n)=1]-Pr[D'(1^n,F_n)=1]|&=\frac{1}{q(n)}|\sum_{i=1}^{q(n)}Pr[D(T_n^{i+1})=1]-Pr[D(T_n^{i})=1]|\\
&=\frac{1}{q(n)}|Pr[D(T_n^{q(n)+1})=1]-Pr[D(T_n^{1})=1]|\\
&>\frac{1}{q(n)p(n)},
\end{align*}
where the last inequality is due to the fact that
$T_n^{q(n)+1}=(1^n,(H_n^1,\ldots,H_n^{q(n)}))$ and
\mbox{$T_n^{1}=(1^n,(F_n^1,\ldots,F_n^{q(n)}))$}. 
But this means that for any such $n$, $D'$ can distinguish
$F=\{F_n\}_{n\in\mathbb{N}}$ and $H=\{H_n\}_{n\in\mathbb{N}}$ with 
non-negligible probability, and thus can do that for infinitely many $n$'s.
This is a contradiction to the assumption that $\{f_s: \{0,1\}^{|s|}\to
\{0,1\}^{|s|}\}_{s\in\{0,1\}^*}$ is a pseudorandom function ensemble. 
\end{proof}

\section{Multi-key Multi-Message Security}
\label{section:appB}
In this section, we show that any secure public-key encryption scheme is
also multi-key multi-message secure.

\begin{lemma}
If $(\Gen,\Enc,\Dec)$ is a secure public key encryption scheme, then it
is also multi-message multi-key secure. 
\end{lemma}
\begin{proof}
Assume for contradiction that 
$(\Gen,\Enc,\Dec)$ is a secure public key encryption scheme that is not
multi-message 
multi-key secure. Then there exist  polynomials
$f$ and $g$ and 
an adversary \mbox{$A=(A_1,A_2)$} such that 
$\{\text{IND-MULT}_0^{\Pi}(A, k,f,g)\}_k$ and $\{\text{IND-MULT}_1^{\Pi}(A, 
k,f,g)\}_k$ are distinguishable. That means there exist a PPT $D$ and a polynomial $p$ such that
$$|Pr[D(1^k,\{\text{IND-MULT}_0^{\Pi}(A, k,f,g)\})=1]-Pr[D(1^k,\{\text{IND-MULT}_1^{\Pi}(A, k,f,g)\})=1]|>\frac{1}{p(n)}.$$

Let $T^{\pi}_{i,j}(A,k,f,g)$ be the following PPT
algorithm: 
\begin{center}
\noindent\begin{tabular}[c]{rlcrl}
$T^{\pi}_{i,j}(A, k,f,g) :=$ & $(pk_1,sk_1)
 \leftarrow \Gen(1^k),\ldots (pk_{g(k)},sk_{g(k)}) \leftarrow \Gen(1^k),$ \\ 
 & $(m_0^1,\ldots,m_0^{f(k)}, m_1^1,\ldots,m_1^{f(k)}, \tau) \leftarrow A_1(1^k, pk_1,\ldots,pk_{g(k)})$\\
 & $\mathcal{C} \leftarrow \Enc_{pk_1}(m_0^1),\ldots, \Enc_{pk_{g(k)}}(m_0^1),$\\ 
&$\ldots,\Enc_{pk_1}(m_0^{j}),\ldots,\Enc_{pk_{i-1}}(m_0^{j}),\Enc_{pk_{i}}(m_1^{j}),\ldots \Enc_{pk_{g(k)}}(m_1^{j}),$\\
&$\ldots\Enc_{pk_1}(m_1^{f(k)}),\ldots, \Enc_{pk_{g(k)}}(m_1^{f(k)})$ \\
 & $o \leftarrow A_2(\mathcal{C},\tau)$ \\
 & $\text{Output } o$.
\end{tabular}
\end{center}

We now define an adversary $A'=(A_1',A_2')$, and show that 
$\{\text{IND}_0^{\Pi}(A', k,f,g)\}_k$ and $\{\text{IND}_1^{\Pi}(A',k,f,g)\}_k$
are not computationally indistinguishable.
$A'_1$ on input $(1^k,pk)$ first chooses \mbox{$i\in\{1,\ldots,g(k)\}$}
uniformly at random. It then generates \mbox{$g(k)-1$} random key 
pairs
$(pk_1,sk_1),\ldots,(pk_{i-1},sk_{i-1}),(pk_{i+1},sk_{i+1}),\ldots,
(pk_{g(k)},sk_{g(k)})$. It then calls $A_1$ with input
$(1^k,pk_1,\ldots,pk_{i-1},pk,pk_{i+1},\ldots,pk_{g(k)})$. After getting
$A_1$'s output \mbox{$M=(m_0^1,\ldots,m_0^{f(k)}, m_1^1,\ldots,m_1^{f(k)},
\tau)$}, $A_1'$ chooses $j\in\{1,\ldots,f(n)\}$ uniformly at random, 
and returns as its output
$(m_0^j,m_1^j,(i,j,pk,pk_1,sk_1,\ldots,pk_{g(k)},sk_{g(k)},M))$. 

$A_2'$ on input
$(\mathcal{C},(i,j,pk,pk_1,sk_1,\ldots,pk_{g(k)},sk_{g(k)},M))$ constructs
input $\mathcal{C}'$ for $A_2$ by first appending the  encryptions of messages
$m_0^1\ldots,m_0^{j-1}$ with all the keys, then appending the encryption of $m_0^j$ with
keys $pk_1,\ldots,pk_i$ and then appends $\mathcal{C}$. It then appends the
encryption of 
$m_1^j$ with keys $pk_{i+2},\ldots,pk_{g(k)}$ and also the encryption of
the messages $m_1^{j+1},\ldots,m_1^{f(k)}$ with each of the keys. It then 
outputs $A_2(\mathcal{C}',\tau)$. 
If $\mathcal{C}$ is the encryption of $m^0_j$ with key $pk$,
then this algorithm is identical to $T^{\pi}_{i+1,j}(A,k,f,g)$ (if
$i=g(k)$ then by $T^{\pi}_{i+1,j}$ we mean $T^{\pi}_{1,j+1}$; we use
similar conventions elsewhere), while if it is the encryption of $m^1_j$
with key  
$pk$, then the algorithm is identical to  $T^{\pi}_{i,j}(A,k,f,g)$. 

We claim that $D$ can distinguish $\{\text{IND}_0^{\Pi}(A',
k,f,g)\}_k$ and $\{\text{IND}_1^{\Pi}(A',k,f,g)\}_k$. 
Note that  $$Pr[D(1^k,\{\text{IND}_0^{\Pi}(A',
k,f,g)\})=1]=\frac{1}{g(k)f(k)}\sum_{j=1}^{f(k)}\sum_{i=1}^{g(k)}Pr[D(1^k,T^{\pi}_{i+1,j}(A,k,f,g))=1]$$
and $$Pr[D(1^k,\{\text{IND}_1^{\Pi}(A',
k,f,g)\})=1]=\frac{1}{g(k)f(k)}\sum_{j=1}^{f(k)}\sum_{i=1}^{g(k)}Pr[D(1^k,T^{\pi}_{i,j}(A,k,f,g))=1].$$ 
Thus,
$$\begin{array}{ll}
&|Pr[D(1^k,\{\text{IND}_0^{\Pi}(A',
k,f,g)\})=1]-Pr[D(1^k,\{\text{IND}_1^{\Pi}(A', k,f,g)\})=1]|\\ 
=&
\frac{1}{g(k)f(k)}|\sum_{j=1}^{f(k)}\sum_{i=1}^{g(k)}(Pr[D(1^k,T^{\pi}_{i+1,j}(A,k,f,g))=1]-Pr[D(1^k,T^{\pi}_{i,j}(A,k,f,g))=1])\\ 
=&\frac{1}{g(k)f(k)}|Pr[D(1^k,\{\text{IND-MULT}_0^{\Pi}(A, k,f,g)\})=1]-
Pr[D(1^k,\{\text{IND-MULT}_1^{\Pi}(A, k,f,g)\})=1]|\\
>&\frac{1}{g(k)f(k)p(k)},
\end{array}$$
where the next-to-last line follows because
\mbox{$T^{\pi}_{1,1}(A, k,f,g)=\text{IND-MULT}_1^{\Pi}(A, k,f,g)$} and
\mbox{$T^{\pi}_{g(k)+1,f(k)}(A, k,f,g)=\text{IND-MULT}_0^{\Pi}(A, k,f,g)$}. Thus, we have a
contradiction to the fact that the encryption scheme is secure. 
\end{proof}

\medskip

\end{document}